\documentclass[runningheads]{llncs}
\usepackage{eurosym}
\usepackage{amssymb}

\usepackage{verbatim}
\usepackage{amsmath}
\usepackage{tikz}
\usepackage{mathrsfs}
\usepackage[margin=10pt,font=small,labelfont=bf]{caption}

\usepackage{graphicx}
\usepackage{dsfont}
\usepackage{enumerate}

\usepackage[title,toc,titletoc,header]{appendix}

\renewcommand{\phi}{\varphi}

\usepackage{tikz}
\usetikzlibrary{calc,matrix,positioning,fit, shapes, automata, arrows, patterns}

\tikzset{
  solid node/.style={circle,draw,inner sep=1.2,fill=black},
  empty node/.style={font=\tiny, circle, outer sep=0pt, inner sep=0pt},
  matrix node/.style={rectangle, inner sep=0pt, draw, minimum height=0.5cm, minimum width=0.8cm, outer sep=0pt},  
  noline/.style={edge from parent/.style={draw=none}},
line/.style={edge from parent/.style={draw}},
}

\usepackage{environ}
\makeatletter
\newdimen\royalignsep@
\def\royalign@preamble{%
   &\hfil
    \strut@
    \setboxz@h{\@lign$\m@th\displaystyle{##}$}%
    \ifmeasuring@\savefieldlength@\fi
    \set@field
    \tabskip\z@skip
   &\setboxz@h{\@lign$\m@th\displaystyle{{}##}$}%
    \ifmeasuring@\savefieldlength@\fi
    \set@field
    \hfil
    \tabskip\royalignsep@
}
\NewEnviron{royalign}[1]{%
  \royalignsep@=#1\let\align@preamble=\royalign@preamble
  \begin{align}\BODY\end{align}}
\NewEnviron{royalign*}[1]{%
  \royalignsep@=#1\let\align@preamble=\royalign@preamble
  \begin{align*}\BODY\end{align*}}
\makeatother

\begin{document}

\title{Stit Semantics for Epistemic Notions Based on Information Disclosure in Interactive Settings}

\author{Aldo Iv\'an Ram\'irez Abarca\inst{1}\and Jan Broersen\inst{2}}
\institute{Utrecht University, Utrecht 3512 JK, The Netherlands \email{a.i.ramirezabarca@uu.nl}\and Utrecht University, Utrecht 3512 JK, The Netherlands \email{J.M.broersen@uu.nl}}

\maketitle

\begin{abstract}

We characterize four types of agentive knowledge using a stit semantics over branching discrete-time structures. These are \emph{ex ante} knowledge, \emph{ex interim} knowledge, \emph{ex post} knowledge, and know-how. The first three are notions that arose from game-theoretical analyses on the stages of information disclosure across the decision making process, and the fourth has gained prominence both in logics of action and in deontic logic as a means to formalize ability. In recent years, logicians in AI have argued that any comprehensive study of responsibility attribution and blameworthiness should include proper treatment of these kinds of knowledge. This paper intends to clarify previous attempts to formalize them in stit logic and to propose alternative interpretations that in our opinion are more akin to the study of responsibility in the stit tradition. The logic we present uses an extension with knowledge operators of the Xstit language, and formulas are evaluated with respect to branching discrete-time models. We also present an axiomatic system for this logic, and address its soundness and completeness.

\end{abstract}

\section{Introduction}

For logicians of action and obligation, there is little debate as to the intuition that agentive responsibility has an important epistemic component. In interactive settings, what agents know before and after a choice of action needs to be taken into account when ascribing responsibility and/or culpability for an outcome. 
Aiming to model the different degrees of culpability according to judicial practice with a stit formalism, Broersen introduced in \cite{broersen2008complete} a logic of `knowingly doing'. Next, a comprehensive study of three kinds of game knowledge (\emph{ex ante, ex interim,} and \emph{ex post}) was provided by Lorini et al. in \cite{lorini2014logical} with the goal of formalizing both responsibility attribution and some attribution-emotions related to it --like guilt or blame. They also used stit for it, just as Horty \& Pacuit recently did in \cite{HortyPacuit2017}, where new epistemic operators are added to basic stit in order to model \emph{ex interim} knowledge and epistemic ability (or know-how). Even more recently, Horty used in \cite{hortyepistemic} the same novel epistemic operators to provide a logic of epistemic obligations, to which \cite{JANANDI} constitutes a reply. Additionally, other attempts to model individual and collective know-how in a stit-like fashion can be found in \cite{naumov2017together}, where the authors take inspiration from Coalition Logic and use transition systems to account for a given group of agents' ability to bring about an outcome by voting.  We mention this in the introduction rather than including it in a section for related work because the present paper has two main objectives, and they both stem from the context of the existing literature in epistemic stit: (1) we want to point out the areas for which the mentioned works overlap, clarifying their differences and shortcomings, and (2) we want to offer new stit formalizations for \emph{know-how}, \emph{ex ante, ex interim}, and \emph{ex post} knowledge, targeting components of them which we believe are essential in an analysis of responsibility. 

Within game theory and epistemic logic, there is a degree of agreement regarding characteristics of the four kinds of knowledge we have mentioned. In broad terms, these characteristics are the following. \textbf{\emph{Ex ante}} knowledge concerns the information that is available to the agents regardless of their choices of action at a given moment. It is commonly thought to be the knowledge that the agents have \emph{before} they choose any of their available actions and execute them. \textbf{\emph{Ex interim} knowledge} can be seen as the knowledge that is private to an agent after choosing an action but before knowing the concurrent choices of other agents. 
\textbf{\emph{Ex post} knowledge} concerns the information that is disclosed to the agents after everybody reveals their choices of action and executes them. \textbf{Know-how} concerns the epistemic ability of bringing about a particular outcome. 

The paper is structured as follows. Since the logic we favor is an extension of Xstit (\cite{xstit}, \cite{Xu2015}), in section 2 we deal with the examples that motivated our choice of syntax and semantics and introduce these two aspects of the logic. In section 3 we present the definitions for \emph{ex ante, ex interim,} \emph{ex post} knowledge, and know-how that we want to formalize and give their characterizations as formulas built with specific combinations of the operators in our language. 
We compare these new interpretations with the previous ones by dissecting the examples from the preceding section. In Section 4 we introduce an axiomatic system for the developed logic and mention its soundness and completeness results, after which we conclude.  

\section{An Example Involving Four Kinds of Game Knowledge}
\label{shit}

We intend to give a formal overview for settings of interactive decisions. Different agents will have different epistemic status and they can act \emph{concurrently}. This means that our models include the layouts of formal games with \emph{incomplete} and \emph{imperfect} information,\footnote{In \cite{aagotnes2015knowledge} \AA gotnes et al. write: ``In game theory, two different terms are traditionally used to indicate
lack of information: `incomplete' and `imperfect' information. Usually, the
former refers to uncertainties about the game structure and rules, while the
latter refers to uncertainties about the history, current state, etc. of the
specific \emph{play of the game.}''} as well as of concurrent game structures (CGS's) (see \cite{aagotnes2015knowledge}, \cite{boudou2018concurrent}) and epistemic transition systems (see \cite{naumov2017together}, \cite{naumov2017coalition}). 
Consider the following example, built as a variation of Horty's puzzles for epistemic ability (\cite{HortyPacuit2017}) and inspired by the film \emph{Mission: Impossible 6}. 

\begin{example}\label{ej1}
A bomb squad consisting of three members ($ethan$, $luther$, and $benji$) faces a complex bomb situation. Terrorists threaten to blow up a facility with two bombs ($L$ and $B$) remotely connected to each other. If the squad defuses one bomb before the other, then the latter is programmed to set off. The terrorists who planted the bombs start a countdown. If the countdown ends both bombs go off. Each bomb has a remote activation system with two main wires, a red one and a green one. Morevoer, each bomb has its own detonator, and these detonators include a fail-safe mechanism that makes it possible for the bombs to be disarmed. The squad figures out that there are \emph{three} ways to successfully defuse the two bombs: (1) If they activate the fail-safe mechanisms of \emph{both} detonators, the squad needs to cut both red wires simultaneously. 
(2) If they manage to activate the fail-safe mechanism for the detonator of bomb $L$ but not of bomb $B$, they need to cut the red wire of $L$ and the green one of $B$ simultaneously.
(3) The reverse situation of the above item in case they only manage to activate the fail-safe mechanism for the detonator of bomb $B$. In other words, cutting the red wire of a bomb without the previous activation of the fail-safe mechanism in its detonator makes it go off. Cutting the green wire \emph{with} previous activation of the fail-safe mechanism also makes it go off. If neither fail-safe mechanism is activated, both bombs go off. If any of these bombs goes off, the explosion is so powerful that it is impossible to ascertain which bomb went off or if both did. 

After the countdown starts, agent $ethan$ is commissioned with the task of retrieving the detonators. Agents $luther$ and $benji$ have to afterwards synchronize the cutting of wires in order to disarm bombs $L$ and $B$, respectively.  A malfunction in the squad's telecom gear causes for them to lose all communication with each other, so $luther$ and $benji$ do not know whether $ethan$ has retrieved one detonator, both of them, or none. Regardless, they know that they need to synchronize the cutting of both bombs' main wires \emph{just before} the countdown ends. 
We explore the following alternatives: \emph{a)} Unbeknownst to $luther$ and $benji$, $ethan$ succeeds in retrieving only the detonator for bomb $B$. They synchronize the cutting of the wires, and since it is statistically better for both $luther$ and $benji$ to cut the two red wires, they do so. Bomb $L$ goes off. \emph{b)} Unbeknownst to $benji$, $luther$ finds out what $ethan$ did. However, $luther$ is actually an undercover associate of the terrorists, so he decides to go on with the cutting of the red wire so that bomb $L$ goes off.
\end{example}


It is convenient to have visual representations of our examples, and stit logic allows us to draw diagrams of them when seen as models. In order to be precise about such models, we introduce the syntax and semantics of the stit logic that we will use before addressing the diagrams. As mentioned before, it is an extension of Xstit logic.\footnote{Xstit was developed by Broersen in \cite{broersen2008complete} on the conceptual assumption that the effects of performing an action are not instantaneous. Rather, an action that is chosen at a given moment will bring about effects only in the next.} 

\begin{definition}[Syntax KXstit]
\label{syntax ep stit}
Given a finite set $Ags$ of agent names and a countable set of propositions $P$ such that $p \in P$ and $\alpha \in Ags$, the grammar for the formal language $\mathcal L_{\textsf{KX}}$ is given by:
\footnotesize
\[ \begin{array}{lcl}
\phi :=  p \mid \neg \phi \mid \phi \wedge \psi \mid \Box \phi \mid X\phi \mid Y\phi \mid [\alpha] \phi \mid [Ags]\phi \mid K_\alpha \phi  
\end{array} \]

\end{definition}
\normalsize
$\Box\varphi$ is meant to express the `historical necessity' of $\varphi$
($\Diamond \varphi$ abbreviates $\neg \Box \neg \varphi$). $X$ is the `next moment' operator and $Y$ is the `last moment' operator. $[\alpha] \varphi$ stands for `$\alpha$ sees to it that $\varphi$'. $[Ags]\varphi$ stands for `the grand coalition $Ags$ sees to it that $\varphi$'. 
$K_\alpha$ is the epistemic operator for $\alpha$. Observe that $\mathcal L_{\textsf{KX}}$ is built with the \emph{instantaneous} action operators $[\alpha]$ and $[Ags]$. We have done so because this expressive language simplifies the axiomatization process for our logic. However, we restrict our treatment of the four kinds of agentive knowledge to `actions that take effect in the next moment', which are characterized by formulas of the form  $[\alpha]X\phi$ and $[Ags]X\phi$. Therefore, we work with a fragment of $\mathcal L_{\textsf{KX}}$. In what follows, we abbreviate the combination $[\alpha]X$ by $[\alpha]^X$ and the combination $[Ags]X$ by $[Ags]^X$ (see \cite{Xu2015} for a similar approach of the matter).   

As for the semantics, the structures with respect to which we evaluate the formulas of $\mathcal L_{\textsf{KX}}$  are based on \emph{epistemic branching discrete-time frames}. 

\begin{definition}[Epistemic branching discrete-time (BDT) frames]
\label{frames}
A tuple of the form $\langle T,\sqsubset, \mathbf{\mathbf{Choice}}, \{\sim_\alpha\}_{\alpha\in Ags}\rangle$ is called an \emph{epistemic branching discrete-time frame} iff
\begin{itemize}
    \item $T$ is a non-empty set of \textnormal{moments} and $\sqsubset$ is a strict partial ordering on $T$ satisfying `no backward branching'. Each maximal $\sqsubset$-chain is called a $\textnormal{history}$, which represents a way in which time  evolves. $H$ denotes the set of all histories, and for each $m\in T$, $H_m:=\{h \in H ;m\in h\}$. Tuples $\langle m,h \rangle$ are called \emph{situations} iff $m \in T$, $h \in H$, and $m\in h$. 
    We call the frames `discrete-time' because $(T,\sqsubset)$ must meet these requirements: \emph{a)} For every $m\in T$ and $h\in H_m$, there exists a unique moment $m^{+h}$ such that $m\sqsubset m^{+h}$ and $m^{+h}\sqsubseteq m'$ for every $m'\in h$ such that $m\sqsubset m'$.
\emph{b)} For every $m\in T$ and $h\in H_m$, there exists a unique moment $m^{-h}$ such that $m^{-h}\sqsubset m$ and $m'\sqsubseteq m^{-h}$ for every $m'\in h$ such that $m'\sqsubset m$.\footnote{Observe that these definitions, coupled with the fact that histories are linearly ordered, imply that for every $m\in T$ and $h\in H_m$, $\left(m^{-h} \right)^{+h}=m$ and $\left(m^{+h} \right)^{-h}=m$.}
  
    \item $\mathbf{\mathbf{Choice}}$ is a function that maps each $\alpha$ and $m$ to a partition $\mathbf{\mathbf{Choice}}^m_\alpha$ of $H_m$. The cells of these partitions represent the actions that are available to each agent at moment $m$. 
    For $m\in T$ and $h\in H_m$ we denote by $\mathbf{\mathbf{Choice}}_{Ags}^m(h)$ the set $\bigcap_{\alpha \in Ags}  \mathbf{\mathbf{Choice}}_\alpha^m(h)$, and we take $\mathbf{\mathbf{Choice}}_{Ags}^m:=\{\mathbf{\mathbf{Choice}}_{Ags}^m(h);h\in H_m\}$, which is the partition of actions available to the grand coalition. $\mathbf{\mathbf{Choice}}$ must satisfy the following constraints. $(\mathtt{NC})$ or `no choice between undivided histories': for all $h, h'\in H_m$, if $m'\in h\cap h'$ for some $m' \sqsupset m$, then $h\in L$ iff $h'\in L$ for every $L\in \mathbf{\mathbf{Choice}}^m_\alpha$. 
$(\mathtt{IA})$ or `independence of agency': a function $s$ on $Ags$ is called a \emph{selection function} at $m$ if it assigns to each $\alpha$ a member of $\mathbf{\mathbf{Choice}}^m_\alpha$. If we denote by $\mathbf{Select}^m$ the set of all selection functions at $m$, then we have that for every $m\in T$ and $s\in\mathbf{Select}^m$, $\bigcap_{\alpha \in Ags} s(\alpha)\neq \emptyset$ (see \cite{belnap01facing} for a discussion of the property).

 \item For each $\alpha\in Ags$, $\sim_\alpha$ is an equivalence relation on the set of \emph{situations}, meant to express the epistemic indistinguishability relation for $\alpha$. At this point, the only condition we impose on these relations is $(\mathtt{NoF})$ or `no forget condition': For every $m\in T$, $h\in H_m$, and $\alpha\in Ags$, if $\langle m^{+h}, h\rangle \sim_\alpha\langle m_*,h_*\rangle$, then $\langle m, h\rangle \sim_\alpha\langle m_*^-,h_*\rangle$. 
\end{itemize}
\end{definition}

\begin{definition}
\label{models KCSTIT}
An epistemic BDT model $\mathcal {M}$ consists of the tuple that results from adding a valuation function $\mathcal{V}$ to an epistemic BDT frame, where $ \mathcal{V}: P\to 
    2^{T \times H}$ assigns to each atomic proposition a set of moment-history pairs.
    Relative to a model $\mathcal{M}$, the semantics for the formulas of $\mathcal {L}_{\textsf{KX}}$ is defined recursively by the following truth conditions, evaluated at a situation $\langle m,h \rangle$:
    \scriptsize
\[ \begin{array}{lll}
\mathcal{M},\langle m,h \rangle \models p & \mbox{iff} & \langle m,h \rangle \in \mathcal{V}(p) \\

\mathcal{M},\langle m,h \rangle \models \neg \phi & \mbox{iff} & \mathcal{M},\langle m,h \rangle \not\models \phi \\

\mathcal{M},\langle m,h \rangle \models \phi \wedge \psi & \mbox{iff} & \mathcal{M},\langle m,h \rangle \models \phi \mbox{ and } \mathcal{M},\langle m,h \rangle \models \psi \\

\mathcal{M},\langle m,h \rangle \models \Box \phi &
\mbox{iff} & \forall h'\in H_m, \mathcal{M},\langle m,h' \rangle \models \phi \\

\mathcal{M},\langle m,h \rangle \models  X \varphi &
\mbox{iff}& \mathcal{M},\langle m ^{+h},h \rangle \models \varphi \\

\mathcal{M},\langle m,h \rangle \models Y \varphi &
\mbox{iff}& \mathcal{M},\langle m ^{-h},h \rangle \models \varphi \\

\mathcal{M}, \langle m,h \rangle \models [\alpha]
\varphi &\mbox{iff}& \forall h'\in \mathbf{\mathbf{Choice}}^m_{\alpha}(h),  \mathcal{M},\langle m, h'\rangle \models \varphi\\

\mathcal{M}, \langle m,h \rangle \models [Ags]
\varphi &\mbox{iff}& \forall h'\in \mathbf{\mathbf{Choice}}^m_{Ags}(h),  \mathcal{M},\langle m, h'\rangle \models \varphi\\

\mathcal{M},\langle m,h \rangle \models K_{\alpha} \phi &
    \mbox{iff} & \forall \langle m',h'\rangle \mbox{ s.t. }  \langle m,h \rangle \sim_{\alpha}\langle m',h' \rangle,\\&&  \mathcal{M},\langle m',h' \rangle \models \phi.\\



\end{array} \]
\normalsize
Satisfiability, validity on a frame, and general validity are defined as usual. We write $|\phi|^m$ to refer to the set $\{h\in H_m ;\mathcal{M},\langle m, h\rangle \models \phi\}$.
\end{definition}
 
With the definitions provided in Definition \ref{models KCSTIT} we can analyze the cases in Example \ref{ej1} as epistemic BDT models. A diagram for Example \ref{ej1} a) is included in Figure \ref{fig1}.

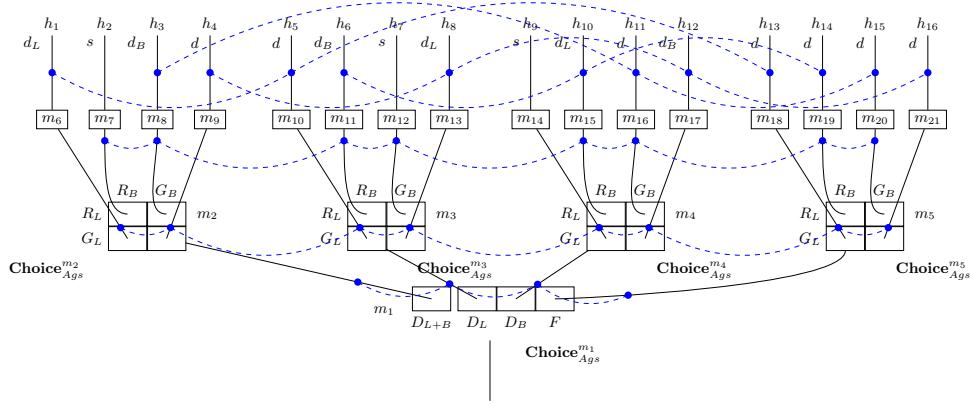
\begin{figure}[htb!]
\begin{minipage}[c]{.9\linewidth}
\centering
\resizebox{370 pt}{!}{%
\begin{tikzpicture}[level distance=2cm,
level 1/.style={sibling distance=1cm},
level 2/.style={sibling distance=5cm},
level 3/.style={sibling distance=1.1cm}
]

\node {} [grow=up]
	child{node [matrix, matrix of nodes, ampersand replacement=\&, label=left:$m_1$, label=below right:$\mathbf{Choice}_{Ags}^{m_1}$] (matrixi) {
	\node[matrix node, label=below:$D_{L+B}$] (m11) {}; \& \node[matrix node, label=below:$D_L$] (m12) {};  \& \node[matrix node, label=below:$D_{B}$] (m13) {};  \& \node[matrix node, label=below:$F$] (m14) {}; \\ }  
	        child[noline]{node [matrix, matrix of nodes, ampersand replacement=\&, label=right:$m_5$, label=below right:$\mathbf{Choice}_{Ags}^{m_5}$]{\node[matrix node, label=above:$R_B$, label=left:$R_L$] (m5_1) {};  \& \node[matrix node, label=above:$G_B$] (m5_2) {} ;\\ \node[matrix node, label=left:$G_L$] (m5_3) {}; \& \node[matrix node] (m5_4) {};\\}
		    child{node [draw,rectangle] (h_12) {$m_{21}$} child {node (m_12) {$h_{16}$}}}
			child{node [draw,rectangle] (h_11) {$m_{20}$}
				child {node (m_11) {$h_{15}$}}}
			child{node [draw,rectangle] (h_10) {$m_{19}$}
				child {node (m_10) {$h_{14}$}}}
            child{node [draw,rectangle] (h_9) {$m_{18}$} child {node (m_9) {$h_{13}$}}}
			}
	      child[noline]{node [matrix, matrix of nodes, ampersand replacement=\&, label=right:$m_4$, label=below right:$\mathbf{Choice}_{Ags}^{m_4}$]{\node[matrix node, label=above:$R_B$, label=left:$R_L$] (m41) {};  \& \node[matrix node, label=above:$G_B$] (m42) {} ;\\ \node[matrix node, label=left:$G_L$] (m43) {}; \& \node[matrix node] (m44) {};\\}
		    child{node [draw,rectangle] (h12) {$m_{17}$} child {node (m-12) {$h_{12}$}}}
			child{node [draw,rectangle] (h11) {$m_{16}$}
				child {node (m-11) {$h_{11}$}}}
			child{node [draw,rectangle] (h10) {$m_{15}$}
				child {node (m-10) {$h_{10}$}}}
            child{node [draw,rectangle] (h9) {$m_{14}$} child {node (m-9) {$h_{9}$}}}
			}
	  	child[noline]{node [matrix, matrix of nodes, ampersand replacement=\&, label=right:$m_3$, label=below right:$\mathbf{Choice}_{Ags}^{m_3}$]{\node[matrix node, label=above:$R_B$, label=left:$R_L$] (m31) {};  \& \node[matrix node, label=above:$G_B$] (m32) {} ;\\ \node[matrix node, label=left:$G_L$] (m33) {}; \& \node[matrix node] (m34) {};\\}
		    child{node [draw,rectangle] (h8) {$m_{13}$}child {node (m-8) {$h_{8}$}}}
			child{node [draw,rectangle] (h7) {$m_{12}$}
				child {node (m-7) {$h_{7}$}}}
			child{node [draw,rectangle] (h6) {$m_{11}$}
				child {node (m-6) {$h_{6}$}}}
            child{node [draw,rectangle] (h5)  {$m_{10}$}child {node (m-5) {$h_{5}$}}}
			}
		child[noline]{node [matrix, matrix of nodes, ampersand replacement=\&, label=right:$m_2$, label=below left:$\mathbf{Choice}_{Ags}^{m_2}$]{\node[matrix node, label=above:$R_B$, label=left:$R_L$] (m21) {};  \& \node[matrix node, label=above:$G_B$] (m22) {} ;\\ \node[matrix node, label=left:$G_L$] (m23) {}; \& \node[matrix node] (m24) {};\\}
		    child{node [draw,rectangle] (h4) {$m_9$}child {node (m-4) {$h_{4}$}}}
			child{node [draw,rectangle] (h3) {$m_8$}
				child {node (m-3) {$h_{3}$}}}
			child{node [draw,rectangle] (h2) {$m_7$}
				child {node (m-2) {$h_{2}$}}}
            child{node [draw,rectangle] (h1) {$m_6$} child {node (m-1) {$h_{1}$}}}
			}	
		};

\draw (m11.center) -- (m24) node[pos=.3,circle,fill=blue, scale=.5] (li1){};
\draw (m12.center) -- (m33) node[pos=.3,circle,fill=blue, scale=.5] (li2){};
\draw (m13.center) -- (m43) node[pos=.3,circle,fill=blue, scale=.5] (li3){};

\draw (m14.center) .. controls +(right:.5cm) and +(down:1cm) .. (m5_3) node[pos=.3,circle,fill=blue, scale=.5] (li4){};

\draw (m21.center) .. controls +(left:.5cm) and +(down:1cm) .. (h2) node[pos=.9,circle,fill=blue, scale=.5] (lh2) {}  node [pos=.5,draw=none] (lh2c) {};
\draw (m23.center) -- (h1) node [pos=.9,draw=none] (lh1) {} node [pos=.1,circle,fill=blue, scale=.5] (lh1b) {} node [pos=.5,draw=none] (lh1c) {};

\draw (m22.center) .. controls +(left:.5cm) and +(down:1cm) .. (h3) node [pos=.9,circle,fill=blue, scale=.5](lh3) {} node [pos=.1,draw=none] (lh3b) {} node [pos=.5,draw=none] (lh3c) {};
\draw (m24.center) -- (h4) node [pos=.9,draw=none] (lh4) {} node [pos=.1,circle,fill=blue, scale=.5] (lh4b) {} node [pos=.5,draw=none] (lh4c) {};
\draw (m31.center) .. controls +(left:.5cm) and +(down:1cm) .. (h6) node [pos=.9,circle,fill=blue, scale=.5] (lh6) {} node [pos=.1,draw=none] (lh6b) {} node [pos=.5,draw=none] (lh6c) {};

\draw (m33.center) -- (h5) node [pos=.9,draw=none] (lh5) {} node [pos=.1,circle,fill=blue, scale=.5] (lh5b) {} node [pos=.5,draw=none] (lh5c) {};

\draw (m32.center) .. controls +(left:.5cm) and +(down:1cm) .. (h7) node [pos=.9,circle,fill=blue, scale=.5](lh7) {} node [pos=.1,draw=none] (lh7b) {} node [pos=.5,draw=none] (lh7c) {};
\draw (m34.center) -- (h8) node [pos=.9,draw=none] (lh8) {} node [pos=.1,circle,fill=blue, scale=.5] (lh8b) {} node [pos=.5,draw=none] (lh8c) {};

\draw (m41.center) .. controls +(left:.5cm) and +(down:1cm) .. (h10)node [pos=.9,circle,fill=blue, scale=.5](lh10) {} node [pos=.1,draw=none] (lh10b) {} node [pos=.5,draw=none] (lh10c) {};
\draw (m43.center) -- (h9) node [pos=.9,draw=none] (lh9) {} node [pos=.1,circle,fill=blue, scale=.5] (lh9b) {} node [pos=.5,draw=none] (lh9c) {};

\draw (m42.center) .. controls +(left:.5cm) and +(down:1cm) .. (h11) node [pos=.9,circle,fill=blue, scale=.5 ](lh11) {} node[pos=.1,draw=none] (lh11b) {} node [pos=.5,draw=none] (lh11c) {};
\draw (m44.center) -- (h12) node [pos=.9,draw=none] (lh12) {} node [pos=.1,circle,fill=blue, scale=.5] (lh12b) {} node [pos=.5,draw=none] (lh12c) {};

\draw (m5_1.center) .. controls +(left:.5cm) and +(down:1cm) .. (h_10)node [pos=.9,circle,fill=blue, scale=.5](lh_10) {} node [pos=.1,draw=none] (lh_10b) {} node [pos=.5,draw=none] (lh_10c) {};
\draw (m5_3.center) -- (h_9) node [pos=.9,draw=none] (lh_9) {} node [pos=.1,circle,fill=blue, scale=.5] (lh_9b) {} node [pos=.5,draw=none] (lh_9c) {};

\draw (m5_2.center) .. controls +(left:.5cm) and +(down:1cm) .. (h_11) node [pos=.9,circle,fill=blue, scale=.5 ](lh_11) {} node[pos=.1,draw=none] (lh_11b) {} node [pos=.5,draw=none] (lh_11c) {};
\draw (m5_4.center) -- (h_12) node [pos=.9,draw=none] (lh_12) {} node [pos=.1,circle,fill=blue, scale=.5] (lh_12b) {} node [pos=.5,draw=none] (lh_12c) {};

\draw[dashed,blue, bend right] (lh2.center) to (lh3.center) {};
\draw[dashed,blue, bend right] (lh3.center) to (lh6.center) {};
\draw[dashed,blue, bend right] (lh6.center) to (lh7.center) {};
\draw[dashed,blue, bend right] (lh7.center) to (lh10.center) {};
\draw[dashed,blue, bend right] (lh10.center) to (lh11.center) {};
\draw[dashed,blue, bend right] (lh11.center) to (lh_10.center) {};
\draw[dashed,blue, bend right] (lh_10.center) to (lh_11.center) {};

\draw[dashed,blue, bend right] (lh1b.center) to (lh4b.center) {};
\draw[dashed,blue, bend right] (lh4b.center) to (lh5b.center) {};
\draw[dashed,blue, bend right] (lh5b.center) to (lh8b.center) {};
\draw[dashed,blue, bend right] (lh8b.center) to (lh9b.center) {};
\draw[dashed,blue, bend right] (lh9b.center) to (lh12b.center) {};
\draw[dashed,blue, bend right] (lh12b.center) to (lh_9b.center) {};
\draw[dashed,blue, bend right] (lh_9b.center) to (lh_12b.center) {};

\draw (h1) -- (m-1) node [pos=.9,draw=none, label=left:\footnotesize$ d_L$] {} node [pos=.5,circle,fill=blue, scale=.5] (m3h1) {};

\draw (h2) -- (m-2) node [pos=.9,draw=none, label=left:\footnotesize$s$] {} ;

\draw (h3) -- (m-3)node [pos=.9,draw=none, label=left:\footnotesize$d_B$] {} node [pos=.5,circle,fill=blue, scale=.5] (m3h3){};

\draw (h4) -- (m-4)node [pos=.9,draw=none, label=left:\footnotesize$d$] {} node [pos=.5,circle,fill=blue, scale=.5] (m3h4){};

\draw (h5) -- (m-5)node [pos=.9,draw=none, label=left:\footnotesize$d$] {} node [pos=.5,circle,fill=blue, scale=.5] (m3h5){};

\draw (h6) -- (m-6)node [pos=.9,draw=none, label=left:\footnotesize$d_B$] {} node [pos=.5,circle,fill=blue, scale=.5] (m3h6){};

\draw (h7) -- (m-7)node [pos=.9,draw=none, label=left:\footnotesize$s$] {};
\draw (h8) -- (m-8)node [pos=.9,draw=none, label=left:\footnotesize$d_L$] {} node [pos=.5,circle,fill=blue, scale=.5] (m3h8){};
\draw (h9) -- (m-9)node [pos=.9,draw=none, label=left:\footnotesize$s$] {};
\draw (h10) -- (m-10)node [pos=.9,draw=none, label=left:\footnotesize$d_L$] {} node [pos=.5,circle,fill=blue, scale=.5] (m3h10){};
\draw (h11) -- (m-11)node [pos=.9,draw=none, label=left:\footnotesize$d$] {} node [pos=.5,circle,fill=blue, scale=.5] (m3h11) {};
\draw (h12) -- (m-12)node [pos=.9,draw=none, label=left:\footnotesize$d_B$] {} node [pos=.5,circle,fill=blue, scale=.5] (m3h12) {};

\draw (h_9) -- (m_9)node [pos=.9,draw=none, label=left:\footnotesize$d$] {}node [pos=.5,circle,fill=blue, scale=.5] (m3h_9){};
\draw (h_10) -- (m_10)node [pos=.9,draw=none, label=left:\footnotesize$d$] {} node [pos=.5,circle,fill=blue, scale=.5] (m3h_10){};
\draw (h_11) -- (m_11)node [pos=.9,draw=none, label=left:\footnotesize$d$] {} node [pos=.5,circle,fill=blue, scale=.5] (m3h_11) {};
\draw (h_12) -- (m_12)node [pos=.9,draw=none, label=left:\footnotesize$d$] {} node [pos=.5,circle,fill=blue, scale=.5] (m3h_12) {};

\draw[dashed,blue, bend right] (li1.center) to (li2.center) {};
\draw[dashed,blue, bend right] (li2.center) to (li3.center) {};
\draw[dashed,blue, bend right] (li3.center) to (li4.center) {};

\draw[dashed,blue, bend right] (m3h1.center) to (m3h5.center) {};

\draw[dashed,blue, bend right] (m3h6.center) to (m3h10.center) {};
\draw[dashed,blue, bend left] (m3h10.center) to (m3h_10.center) {};

\draw[dashed,blue, bend left] (m3h3.center) to (m3h11.center) {};

\draw[dashed,blue, bend right] (m3h11.center) to (m3h_11.center) {};

\draw[dashed,blue, bend right] (m3h4.center) to (m3h8.center) {};
\draw[dashed,blue, bend left] (m3h8.center) to (m3h12.center) {};
\draw[dashed,blue, bend right] (m3h12.center) to (m3h_12.center) {};

\draw[dashed,blue, bend left] (m3h5.center) to (m3h_9.center) {};

\end{tikzpicture}
}
\captionof{figure}{Example \ref{ej1} a) with epistemic status of $luther$}
\label{fig1}
\end{minipage}
\end{figure}

 The diagram represents the different possibilities in which time may evolve from the point the bomb squad sets out to defuse the bombs onward. Each history stands for one of these relevant possibilities, according to the actions taken by the agents. We take $d_L$ and $d_B$ to denote the propositions  `bomb L detonates' and `bomb B detonates', respectively, and we abbreviate $d_L\land d_B$ with $d$. We take $s$ to denote the proposition `the bombs are defused'. These are true or false depending on the moment/history pair of evaluation. For instance, at situation $\langle m_2, h_4\rangle$ the bombs have not detonated ($\mathcal{M},\langle m_2,h_4 \rangle \models \lnot d_L\land \lnot d_B$), but in the next moment they have ($\mathcal{M},\langle m_9,h_4 \rangle \models d$). 
For clarity, we have labeled the actions available to the agents within the choice partitions of each moment. Label $D_{L+B}$ stands for  `activating the fail-safe mechanism of both detonators' (similarly for $D_{L}$ and $D_{B}$), $F$ stands for `failing to secure a detonator', $R_L$ ($R_B$) stands for `cutting the red wire of bomb $L$ $(B)$', and $G_L$ ($G_B$) stands for `cutting the green wire of bomb $L$  $(B)$'. In epistemic BDT models all agents in $Ags$ get to choose from their `available' actions at \emph{every} moment/history pair. In our example, $Ags$ is made up of $ethan$,  $luther$ and $benji$. However, $luther$ and $benji$ cannot choose anything at moment $m_1$, so we take it that their available actions both lie within the trivial partition, meaning that $\mathbf{\mathbf{Choice}}_{luther}^{m_1}= \mathbf{\mathbf{Choice}}_{benji}^{m_1}=\{H_{m_1}\}$. Similarly, we have that $\mathbf{Choice}_{ethan}^{m_i}=\{H_{m_i}\}$ for $m_2$, $m_3$, $m_4$, and $m_5$.

In both cases a) and b) of Example \ref{ej1}, at moment $m_1$ $ethan$ `chooses' an action and executes it with effects in the next moment. The set of `next moments' for $m_1$ comprises  $m_2$, $m_3$, $m_4$, and $m_5$. The histories running through each of these are partitioned according to the choices of action available to $luther$ and $benji$. The frame condition \emph{independence of agency} makes the partitions-layout look like a game in normal form, modelling concurrent choice of action. Observe that although this example presupposes that there is a single next state per action profile, this need not be the case for the general models.\footnote{Models that satisfy the condition that for every $m\in T$ and $h\in H_m$, if $h'\in \mathbf{Choice}_{Ags}^m(h)$ then $m^{+h}=m^{+h'}$ are called \emph{deterministic} (see \cite{xstit}, \cite{van2017doing}).} 

In both cases a) and b) of Example \ref{ej1}, at moments  $m_2$, $m_3$, $m_4$, and $m_5$ the available actions for $luther$ are either $R_L$ or $G_L$, and the available actions for $benji$ are either $R_B$ or $G_B$. A given outcome will ensue according to which action profile is chosen by $Ags$ --at these moments, only $luther$ and $benji$'s choices are relevant. For example, let us suppose that at moment $m_4$, $benji$ cuts the red wire (he chooses $R_B$), and $luther$ chooses to cut the green one (he chooses $G_L$). This means that we constrain $H_{m_4}$ to $\{h_9\}$, where the bombs are defused. In this case, the semantics in definition \ref{models KCSTIT} implies that $\mathcal{M},\langle m_4,h_9 \rangle \models [Ags]^X s$. As examples of formulas involving traditional Xstit operators, we have that $\mathcal{M},\langle m_1,h_2 \rangle \models X(\lnot d_L\land\lnot d_B)$, that $\mathcal{M},\langle m_1,h_{10} \rangle \models X[luther]^X d_L$, that $\mathcal{M},\langle m_4,h_{11} \rangle \not\models Y[Ags]^X d$, and that $\mathcal{M},\langle m_3,h_{7} \rangle \models Y\Diamond X[luther]^X d_L$. 

In this paper we are concerned with the epistemic structure of the agents and with what it says about \emph{(a)} their knowledge through the different stages of choice of action, and \emph{(b)} what they are able to knowingly do. The epistemic structure is given by the indistinguishability relations. In Figure 1 we represent $luther$'s indistinguishability relation in Example 1 a) with blue dashed lines. We focus our study on the epistemic structures of $luther$ and $benji$, since it is only because of them that cases a) and b) of Example \ref{ej1} are different. Although $ethan$ is also endowed with an indistinguishability relation over the set of situations in our BDT structure, we omit its analysis.                           
The explanation is the following.  At moment $m_1$ $luther$ (and $benji$, for that matter) does not know what $ethan$ does. At moments $m_2, m_3$, $m_4$, and $m_5$, we observe that $luther$ is able to distinguish between either cutting the red or the green wire of his bomb. If we take the proposition $r_L$ to mean `the red cable has been cut in bomb $L$', then we have that $\mathcal{M},\langle m_i,h \rangle \models  K_{luther}[luther]^X r_L \lor  K_{luther}[luther]^X \lnot r_L$ for $i\in \{2,3,4,5\}$ and $h\in H_{m_i}$.\footnote{This is related to a property that epistemic game theorists call ``knowledge of one's own action'', a feature which we will address in section 3, when we compare the different interpretations of knowledge and their implications. Although it holds in this example, we do not enforce it in all BDT frames.} So \emph{in this example} we assume that $luther$ (and $benji$) knowingly performs his available actions in the sense of knowingly performing the actions labeled by $R_L$ and $G_L$. However --and this is essential in our treatment of knowledge-- this is not mandatory for all our frames. Observe that our example still accounts for the possibility that agents bring about certain outcomes without knowingly doing so. In Example \ref{ej1} a), we have that even if $\mathcal{M},\langle m_4,h_{10} \rangle \models [luther]^X d_L$, it is the case that $\mathcal{M},\langle m_4,h_{10} \rangle \models \square \lnot K_{[luther]}[luther]^X d_L$ (there is no way in which $luther$ knowingly sets a bomb off). At moments $m_6$--$m_{21}$, although the fate of the bombs has been decided, agents may still be uncertain about the exact cause of the explosion in the cases where $luther$ and $benji$ have not both cut the green wire of their respective bomb. We use the epistemic structures of the agents at these moments in order to illustrate instances of \emph{ex post} knowledge. Observe, for example, that $\mathcal{M},\langle m_{11},h_6 \rangle \models \lnot K_{luther}Y [benji]^X d_B$.

\section{\emph{Ex Ante}, \emph{Ex interim}, \emph{Ex Post}, and Know-how}

Recent trends in modelling responsibility by way of `knowingly doing' and epistemic ability (`know-how') base these two notions on the differential knowledge across the stages of decision making. In this section we describe diverse interpretations that authors have given in the past to the four kinds of knowledge, and we compare them with our own versions. As mentioned in the introduction, we have two goals in mind. One is to clarify overlapping intuitions for the work done in epistemic stit and address the shortcomings of previous analyses. The other is to present new formal characterizations that we believe are more akin to modelling responsibility with stit. 

As we see it, we are modelling agents' \emph{uncertainty} in strategic interaction. The branching-time theory and language of action that lies at the heart of stit allows us to do this neatly, being flexible enough to address different angles of agentive uncertainty. For this work we discern the following levels of uncertainty, in clear correlation with the four kinds of knowledge: \emph{(a)} \emph{Uncertainty about previous actions}.
   \emph{(b)} \emph{Uncertainty about the nature and effects of one's own actions}.
    \emph{(c)} \emph{Uncertainty about other agents' actions}.
    \emph{(d)} \emph{Uncertainty about the effects of joint actions.}
The essence of our criticism to previous proposals also underlies what we strive for the most in \emph{our} semantics:  \emph{flexibility}. We want for the epistemic stit models to be flexible enough to address both coarse- and fine-grained kinds of uncertainty in interactive settings.  

\textbf{\emph{Ex ante} knowledge}. It is commonly thought of as the kind of knowledge that an agent has \emph{regardless} of its choice of action at a given moment (see \cite{lorini2014logical}, \cite{aumann2008rational}). Previous renditions of the concept in epistemic stit all try to model this quality, but from somewhat different viewpoints. To illustrate the knowledge that \emph{we} intend to formalize, consider our Example \ref{ej1}. In case a), at situation $\langle m_4, h_{10}\rangle$ --which we have taken as an actual situation-- neither $luther$ nor $benji$ has \emph{ex ante} knowledge that $ethan$ activated the fail-safe mechanism in the detonator for bomb $B$. In case b), on the contrary, $luther$ does know it \emph{ex ante}. Therefore, we favor the view that if an agent has certainty about previous actions, it is easier for that agent to discern things $\emph{ex ante}$. 

 Lorini et al. present in \cite{lorini2014logical} an epistemic stit logic with three operators for \emph{ex ante, ex interim,} and \emph{ex post} knowledge: $K_\alpha^{\bullet \circ \circ}, K_\alpha^{\circ \bullet\circ},$ and $K_\alpha^{\circ \circ \bullet}$, respectively. They base the semantics of all three of them on primitive epistemic indistinguishability relations for \emph{ex ante} knowledge. The respective intersections of these relations with those of individual and collective action --themselves resulting from a structure of action labels-- yield then Lorini et al.'s versions of \emph{ex interim} and \emph{ex post} knowledge. Their treatment follows game theory's natural assumption that \emph{ex interim} knowledge refines \emph{ex ante}, and that \emph{ex post} knowledge in turn refines \emph{ex interim}. We can safely suppose that this is the reason why they take the semantics for \emph{ex ante} knowledge as the starting point of all three. Although the main problems with Lorini et al.'s system will be discussed when we deal with \emph{ex interim} knowledge, we mention here that the fact that the authors do not enforce any connection between \emph{ex ante} knowledge and historical necessity poses a problem for identifying what an agent knows \emph{ex ante} with the knowledge that is present regardless of that agent's choice of action. In other words, Lorini et al.'s system admits situations in which at a given moment and along a given history an agent knows $\phi$ \emph{ex ante}, but if the agent `changes' its choice of action then it would stop knowing $\phi$ (in Lorini et al.'s logic, $K_\alpha^{\bullet \circ \circ}\phi \to \square K_\alpha^{\bullet \circ \circ}\phi$ is not valid).\footnote{In \cite{lorini2014logical} (p.1320), Remark 2.6 actually states that the authors do not intend for agents to consider all their available choices epistemically \emph{possible} in the \emph{ex ante} sense, but this leads precisely into having choice-dependent \emph{ex ante} knowledge.}

In \cite{HortyPacuit2017} and \cite{herzig2006knowing}, Horty \& Pacuit and Herzig \& Troquard, respectively, tackle a notion of \emph{ex ante} knowledge with epistemic stit in similar ways. 
It is only under the light of their full systems that we find fault in both approaches, which we address when dealing with \emph{ex interim} knowledge and know-how.  

\textbf{Our version:} We take $\alpha$'s \emph{ex ante} knowledge to be all the truths \emph{about the next moment} that, regardless of its current choice of action, $\alpha$ knows to be independent of that choice of action. Let $\mathcal{M}=\langle T,\sqsubset, \mathbf{\mathbf{Choice}}, \{\sim_\alpha\}_{\alpha\in Ags}\rangle$ be an \emph{epistemic BDT frame} and $\phi$ of $\mathcal L_{\textsf{KX}}$. We say that at situation $\langle m,h \rangle$, $\alpha$ has \emph{ex ante} knowledge of $\phi$ iff $\mathcal{M},\langle m, h \rangle\models \square K_\alpha \square X \phi$. For instance, we can see that in our Example \ref{ej1}, if we again take $f_B$ to denote the proposition `the fail-safe mechanism of bomb $B$ has been activated', we get that $\mathcal{M},\langle m_4, h_{10}\rangle \not\models \square K_{luther} \square X Y f_B$ for case a), and that $\mathcal{M},\langle m_4, h_{10}\rangle \models \square K_{luther} \square X Y f_B$ for case b).

Two points must be made. First, we observe that we have modelled an \emph{individual} version of \emph{ex ante} knowledge, with no explicit mention of whether there is a collective counterpart or whether \emph{ex ante} knowledge is already collective in nature from the start. There is a sense in which \emph{ex ante} knowledge has been thought as strictly non-private information. Aumann \& Dreze, for instance, state that at the \emph{ex ante} stage of differential information environments, no agent should have any private knowledge (\cite{aumann2008rational}, p. 80). This entails that \emph{ex ante} knowledge should then be seen as an instance of group knowledge. Above we studied an individual version, and we then model Aumann \& Dreze's \emph{ex ante} knowledge with $\square C \square X \phi$, where $C$ stands for the operator for common knowledge for the grand coalition $Ags$ in the language $L_{\textsf{KX}}$. Second, Duijf points out in \cite{duijf2018let} (Chapter 3) that it is important to account for the temporal dimension of \emph{ex ante} knowledge, where this is seen as the knowledge that agents have \emph{before} they and the others engage in choices of action. Our account does acknowledge such temporal dimension insofar it is given within Xstit logic, and thus the information that agents have \emph{ex ante} is available \emph{before} the effects of choices of action take place. 

\textbf{\emph{Ex interim} knowledge}. It is assumed that at the \emph{ex interim} stage of decision making, an agent's available information expands its \emph{ex ante} knowledge by taking into account \emph{its} choice of action, though not yet the other agents' ones. If an agent knows $\phi$ \emph{ex interim}, then the agent is certain about the fact that $\phi$ will occur after the performance of its choice of action and regardless of what other agents choose. For instance, in Example \ref{ej1} b) $luther$ knows \emph{ex interim} that bomb $L$ will go off after he chooses $R_L$. We side by the intuition that if an agent has some certainty about the effects of its own actions, then it is easier for that agent to discern things \emph{ex interim}.\footnote{We still want to allow for situations in which agents do not always have certainty about the actions that they choose even after choosing them. As will be seen later, this distinguishes our interpretation of \emph{ex interim} knowledge from epistemic game theory's traditional one.} 

For Lorini et al., the information available to an agent \emph{ex interim} is not only independent of the other agents' choices of action, but also \emph{must} make agents discern which action they are taking. This means that in Lorini et al.'s formalization agents will always be able to know \emph{ex interim} the action they perform and can never be uncertain at the \emph{ex interim} stage about the difference between actions that have different labels. The condition can be expressed --in the terminology of \cite{lorini2014logical}-- by use of propositional constants representing the execution of action labels: if for every action label $A$ and $\alpha$ we take $p_A^\alpha$ to denote the proposition `The action $A$ is performed by $\alpha$', then the following formula is valid in Lorini et al.'s system: $p_A^\alpha\to K^{\circ\bullet\circ}p_A^\alpha$. Although the property is to a certain extent in keeping with epistemic game theory, we find it constraining in the context of epistemic stit. For instance, in Example \ref{ej1} we may want to model a situation in which $ethan$ does not know the difference between detonator $L$ and detonator $B$ but still activates the fail-safe mechanism of one of them. We would like to say that in this case $ethan$ is not able to discern \emph{ex interim} whether he chose to activate the mechanism for bomb $L$ or for bomb $B$, but according to Lorini et al.'s formalization this would be impossible. 
The constraint is all the more problematic because of its consequences for Lorini et al.'s treatment of responsibility/culpability attribution: in their formalism, agents will always be morally responsible for performing a given action, without being able to excuse themselves from moral responsibility by claiming that they were uncertain about which action they chose. Observe also that the constraint implies that \emph{ex ante} certainty of the present moment forces agents to know all the consequences of their actions, so options for modelling excusability are restricted.   

In \cite{HortyPacuit2017} Horty \& Pacuit work with similar ideas to Lorini et al.'s when it comes to the relationship between \emph{ex interim} knowledge and action labels. With the goal of disambiguating the sense of epistemic ability from that of causal ability in stit, they base a version of know-how on novel semantics for \emph{ex interim} knowledge. In order to deal with both \emph{ex interim} knowledge and know-how,  Horty \& Pacuit extend basic stit with syntactic and semantic components. Syntactically, they introduce the operator $[\dots kstit]$ to encode agentive \emph{ex interim} knowledge and add it to a language that includes operators $\mathcal{K}_\alpha$ for epistemic indistinguishability, operators $[\alpha \ stit]$ for the traditional Chellas-stit notion of action, and $\square$ for historical necessity. The semantics for formulas involving $[\dots kstit]$ uses \emph{action types} under the premise that actions of the same \emph{type} might lead to different outcomes in different moments. At a given moment and along a given history, $\alpha$ knowingly sees to it that $\phi$ iff at all moments that are epistemically indistinguishable for $\alpha$ to the one of evaluation, $\alpha$'s execution of the same action type enforces $\phi$.\footnote{Formally, $\mathcal{M}, \langle m,h\rangle \models [\alpha \ kstit]\phi$ iff for every $\langle m',h'\rangle$ such that $\langle m, h\rangle \sim_\alpha \langle m',h'\rangle$, $Exn_{m'}(Lbl^\alpha(\langle m,h\rangle))\subseteq |\phi|^{m'}$, where $Lbl_\alpha$ is a function that maps a situation to the action type of the action token being performed at that situation, and $Exn_{m'}$ is a partial function that maps types to their corresponding tokens at moment $m'$. } As mentioned in \cite{JANANDI}, the use of types brings three unfavorable constraints: (1) In order for $[\dots kstit]$ to be an \textbf{S5} operator, the primitive epistemic indistinguishability relations --those supporting the operators $\mathcal{K}_\alpha$-- must ensue not between moment-history pairs but between moments. This limits the class of models to those in which knowledge is moment-dependent and for which agents will not be able to know that they perform a given action.\footnote{Horty \& Pacuit's models satisfy the following constraint: if $\langle m, h\rangle \sim_\alpha \langle m',h'\rangle$, then $\langle m,h_*\rangle \sim_\alpha \langle m',h'_*\rangle$ for every $h_*\in H_m$, $h'_*\in H_{m'}$, which, under reflexivity of $\sim_\alpha$, corresponds syntactically to the axiom schema $\mathcal{K}_\alpha\phi \to \square \phi$.} This leads Horty \& Pacuit into identifying $\mathcal{K}_\alpha \phi$ with $\alpha$'s \emph{ex ante} knowledge of $\phi$, but then a shortcoming is that both instances of knowledge ($\mathcal{K}_\alpha$ and $[\alpha \ kstit]$) satisfy what we call the `own action condition' $(\mathtt{OAC})$. This condition is semantically stated by the following rule: for every situation $\langle m_*,h_*\rangle$, if $\langle m_*, h_*\rangle\sim_\alpha \langle m, h\rangle$ for some $\langle m,h\rangle$, then $\langle m_*,h_*'\rangle\sim_\alpha \langle m,h\rangle$ for every $h_*'\in \mathbf{Choice}^{m_*}_\alpha (h_*)$. It implies the validity of the formulas $\mathcal{K}_\alpha \phi \to [\alpha \ stit]\phi$ and $[\alpha \ kstit]\phi\to [\alpha \ stit]\phi$. In Horty \& Pacuit's formalism, then, there is no sense whatsoever in which agents can know more than what they bring about.\footnote{To see how this constraint thwarts an analysis of the interaction between knowledge and action, consider our Example \ref{ej1}, and assume that at moment $m_4$, for instance, we want to say that $luther$ knows \emph{in some sense} --not in an \emph{ex ante} or \emph{ex interim} sense, though-- what $benji$ will choose. Therefore, $luther$ should in principle be able to somehow distinguish $h_9$ from $h_{12}$, and $h_{10}$ from $h_{11}$. However, in presence of $(\mathtt{OAC})$, this cannot be the case (see \cite{duijf2018let} (Chapter 3) for a more elaborate discussion about the undesirability of this property in epistemic stit).}  (2) Just as with Lorini et al., the semantics for $[\dots kstit]$ entails that agents cannot have uncertainty about the actions they perform at the \emph{ex interim} stage. (3) Indistinguishable moments must offer the same available types.\footnote{In all the treatments of game knowledge that we presently review, virtually nobody disagrees with some version of this constraint. In the case of the approaches from game theory, ATL, and Coalition Logic, the premise is very much related to the concept of `uniform strategies'. We will address it further when analyzing versions of know-how. It is worth mentioning that Lorini et al. remain vague about the subject. The examples they present all presuppose the condition, but they do not demand it explicitly or even refer to it. If they were to enforce it, it would bring the same technical problem as in Horty \& Pacuit's \cite{HortyPacuit2017}.} The problem with this constraint, which we call `uniformity of available action types' ($\mathtt{UAAT}$), is technical in nature: it cannot be characterized syntactically \emph{in their logic} without producing an infinite axiomatization. This is due to the fact that performing a certain action type can only be expressed syntactically with propositional constants.

In \cite{herzig2006knowing} Herzig \& Troquard present a version of \emph{ex interim} knowledge under the term \emph{dynamic knowledge} in order to give semantics for know-how, as Horty \& Pacuit. They extend the language of basic Xstit with a modal operator $[\dots Kstit]$ for this dynamic knowledge. Contrary to Lorini et al. and Horty \& Pacuit, it is for the semantics of formulas involving \emph{this} operator that they use primitive epistemic indisinguishability relations, and they define the truth conditions \emph{both} for action ($[\alpha \ Stit]\phi$) and for static knowledge ($\square [\alpha \ Kstit]\phi$) in terms of it. Still, their formalization ends up working in a similar way to Horty \& Pacuit's. First, their versions of knowledge --both static and dynamic-- satisfy $(\mathtt{OAC})$, so again there is no sense in which agents can know more than what they bring about. Second, they restrict to situations for which agents cannot be uncertain about the actions they choose. Lastly, they favor a condition of uniformity that corresponds to Horty \& Pacuit's $(\mathtt{UAAT})$. 

In \cite{duijf2018let}, Duijf has an interesting proposal that somewhat resembles Herzig \& Troquard's. Duijf thinks that a primitive indistinguishability relation that links moment-history pairs in instantaneous stit characterizes a kind of knowledge that we may already call \emph{ex interim}. 
What distinguishes his interpretation from Herzig \& Troquard's --as well as from Horty \& Pacuit's, for that matter-- is that Duijf's \emph{ex interim} knowledge is flexible enough to deal both with cases of uncertainty about one's action (it is not the case that agents \emph{must} know their actions \emph{ex interim}) --which accounts for possibility of a more coarse-grained knowledge-- and with cases for which agents know more than what they bring about --which accounts for instances of fine-grained knowledge.\footnote{Duijf does not demand for his  models to validate $(\mathtt{OAC})$. However, as we will point out when addressing his formalization of know-how, Duijf does enforce a constraint corresponding to Horty \& Pacuit's $(\mathtt{UAAT})$.} However, it may be a bit too flexible, for it allows situations in which an agent knows \emph{ex interim} that another agent is bringing about something that is not settled, which is clearly dependent on the other agent's choice. Denoting Duijf's operator for \emph{ex interim} knowledge by $\mathds{K}_\alpha$ and the traditional Chellas-stit operators by $[\alpha]$, the formula $\mathds{K}_\alpha [\beta]\phi \land \lnot \square \phi$ is satisfiable in Duijf's logic. It might be surprising that Lorini et al.'s version of \emph{ex interim} knowledge also does not satisfy $(\mathtt{OAC})$, so a comparable criticism can be advanced. In fact, their formalism is such that agents could have \emph{ex ante} knowledge of what other agents are bringing about, something atypical. 

\textbf{Our version:} We identify \emph{ex interim} knowledge with Broersen's notion of `knowingly doing' (\cite{xstit}), so that an agent knows $\phi$ \emph{ex interim} iff it knowingly sees to it that $\phi$ happens in the next moment. Thus, let $\mathcal{M}$ be an \emph{epistemic BDT model} and $\phi$ of $\mathcal L_{\textsf{KX}}$. We say that at situation $\langle m,h \rangle$, $\alpha$ has \emph{ex interim} knowledge of $\phi$ iff $\mathcal{M},\langle m, h \rangle\models K_\alpha [\alpha]^X\phi$. In Example \ref{ej1} b) we get that $luther$ knowingly sees to it that bomb $L$ goes off, so that he knows \emph{ex interim} that bomb $L$ will go off ($\mathcal{M}, \langle m_4, h_{10}\rangle \models K_{luther}[luther]^Xd_L$.) In contrast, in case a) we have that $\mathcal{M}, \langle m_4, h_{10}\rangle \models \lnot K_{luther}[luther]^Xd_L$, so that $luther$ in fact sets bomb $L$ off, but not knowingly and thus without \emph{ex interim} knowledge about it. As for how our notion of \emph{ex interim} knowledge deals with the constraints that we have criticized in the reviewed literature, we advance two important remarks: \begin{enumerate}
\item We do not impose any condition on the primitive epistemic indistinguishability relations that would exclude uncertainty of one's own action in the \emph{ex interim} stage, so we allow for cases of coarse-grained knowledge that have clear implications for responsibility/culpability attribution in stit logic. In particular, full certainty about the moment of evaluation does not imply that agents will know the effects of all their current actions. 
    \item Our system is flexible enough so that agents can \emph{in some sense} know more than what they bring about. Since we do not impose $(\mathtt{OAC})$ for our traditional knowledge semantics, we account for instances of finer-grained knowledge. However, agents \emph{cannot} know \emph{ex interim} more than what they bring about (reflexivity of the primitive indistinguishability relations entails that  $K_\alpha [\alpha]^X\phi\to [\alpha]^X\phi$ is valid with respect to the class of our frames).\footnote{Comparing our version with Duijf's, we observe that we exclude situations where agents can know \emph{ex interim} that other agents will see to it that $\phi$ in the next moment without it being settled that  $\phi$ will hold in the next moment, while we had seen he does not. $K_\alpha [\alpha]^X Y [\beta]^X\phi\to \square X\phi$ is valid with respect to the class of our frames.} 
\end{enumerate}
\textbf{\emph{Ex post} knowledge}. At the last stage of information disclosure, it is revealed to the agents which choices of action they engaged in. Game theorists call the knowledge that arises at this point \emph{ex post} knowledge. This is commonly seen as including facts that hold after constraining the possible histories to the epistemic equivalents of the choice profile of the grand coalition.  Although Aumann \& Dreze assume that at the \emph{ex post} stage ``all information is revealed to all'' (\cite{aumann2008rational}, p.80), we take it to mean \emph{all information ensuing from disclosing the agents' choices}. Game theorists also describe \emph{ex post} knowledge as the kind of knowledge that can be attained \emph{after} all the agents have performed their `strategies', which adds a temporal dimension to the concept and thus increases its ambiguity. In our view, if the agents have some certainty about the effects of joint action, then it is easier for them to know things \emph{ex post}. In Example \ref{ej1} a), we consider as instances of \emph{ex post} knowledge at the actual situation the facts that $Ags$ caused a bomb to go off and did so \emph{unknowingly}. We do not consider that that it should also be \emph{ex post} knowledge the fact that actually it was $luther$ who set off the bomb and that it was bomb $L$. 

Out of all the papers that we have reviewed so far, only \cite{lorini2014logical} includes some treatment for \emph{ex post} knowledge in epistemic stit. There, the relation that provides semantics for formulas involving $K_\alpha^{\circ \circ \bullet}$ is built by intersecting the primitive relation for \emph{ex ante} knowledge with the relation for collective action of the grand coalition --itself the intersection of the agents' relations for individual action (see Definitions \ref{frames} and \ref{models KCSTIT} for the corresponding semantics of $[Ags]^X\phi$). A point to be made is that the instantaneous nature of their action semantics fails to make allowance for an analysis of the temporal dimension of this kind of knowledge. For instance, their semantics does not admit situations for which an agent knows \emph{ex post} that it brought about $\phi$ without the agent also knowing \emph{ex interim} that it brought about $\phi$ ($K_\alpha^{\circ \circ \bullet}[\alpha \ \mathbf{stit}]\phi\to K_\alpha^{\circ \bullet\circ}[\alpha \ \mathbf{stit}] \phi$ turns out to be valid with respect to their frames, where $[\alpha \ \mathbf{stit}]$ are the traditional Chellas-stit operators given in their notation). 

\textbf{Our version(s)}: In our semantics, we propose a version for individual \emph{ex post} knowledge in the following way. Let $\mathcal{M}$ be an \emph{epistemic BDT frame} and $\phi$ of $\mathcal L_{\textsf{KX}}$. We say that at situation $\langle m,h \rangle$, $\phi$ is \emph{ex post} knowledge of $\alpha$ iff $\mathcal{M},\langle m, h \rangle\models X \ K_\alpha \ Y \ [Ags]^X\phi$. In this way, in Example \ref{ej1} a) we get that $\mathcal{M},\langle m_4, h_{10} \rangle\models X \ K_{luther/benji} \ Y \ [Ags]^X d_L\lor d_B$. For case b) we get that $\mathcal{M},\langle m_4, h_{10} \rangle\models  \lnot X \ K_{benji} \ Y \ [Ags]^X (Y [luther]^Xd_L)$ ($benji$ does not know even \emph{ex post} that $luther$ set off bomb $L$).\footnote{Other interesting cases appear in the non-actual situations where the bombs were defused. In Example \ref{ej1} a), for instance, if we take $f_B$ to denote the proposition `the fail-safe mechanism of bomb $B$ has been activated', then $\mathcal{M},\langle m_4, h_9 \rangle\models \lnot K_{luther/benji} (Y [ethan]^X f_B)$ and $\mathcal{M},\langle m_4, h_9 \rangle\models  X \ K_{luther/benji} \ Y \ [Ags]^X (Y \ Y [ethan]^X f_B)$ ($luther$ and $benji$ \emph{realize} \emph{ex post} that $ethan$ secured the detonator for bomb $B$). Observe that, contrary to Lorini et al.'s formalization, ours does account for cases where an agent knows \emph{ex post} that it brought about $\phi$ without knowing \emph{ex interim} that it would bring about $\phi$: $ X \ K_\alpha \ Y \ [Ags]^X Y[\alpha]^X \phi\to K_\alpha [\alpha]^X \phi$ is not valid with respect to our frames.}


Endowed with our three versions of differential knowledge according to the stages of decision making, we can see how they interact with each other. Complying with the customary game theoretical view, we have that \emph{ex post} refines \emph{ex interim}, which in turn refines \emph{ex ante.}

\begin{proposition}
Let $\mathcal{M}$ be an \emph{epistemic branching discrete-time frame}, $\phi$ of $\mathcal L_{\textsf{KX}}$, and $\alpha\in Ags$. We have that
$\mathcal{M}\models \square K_\alpha \square X \phi \to K_\alpha [\alpha]^X \phi$ and $\mathcal{M}\models K_\alpha [\alpha]^X \phi \to X \ K_\alpha \ Y \ [Ags]^X\phi$. 

\end{proposition}

\textbf{Know-how}. When we talk about know-how we refer to the so-called \emph{practical} or \emph{procedural knowledge} of an agent that takes `actions' rather than propositions as content (see \cite{duijf2018let}, Chapter 3).\footnote{According to \cite{duijf2018let}, in \cite{fantl2008knowing} Fantl draws the outlines of know-how by distinguishing it from two other kinds of knowledge: \emph{knowledge by acquaintance} and \emph{propositional knowledge} (know-that). Setting aside for the moment the concept of knowledge by acquaintance, \cite{duijf2018let} proposes that the essential difference between know-how and know-that lies in the content they take. We side by this interpretation, where procedural knowledge concerns actions, and propositional knowledge concerns propositions. Such a disambiguation identifies the concept of know-how that we study with Wang's \emph{goal-directed} know-how. This is related to the debate introduced by Ryle (\cite{ryle2009concept}) as to whether know-how can be reduced to know-that or not, where \emph{intellectualists} think it can be reduced and \emph{anti-intellectualists} think it cannot.} The intuition is that an agent knows how to do something iff it has the procedural knowledge of bringing about that something. We are not engaging in a circular argument here, for the second statement can be described with precise definitions for `knowledge', `action', and `possibility'. Still, we acknowledge that there is a lively debate in the literature as to what being able to bring about a certain outcome exactly means, and the question of what it means to be \emph{epistemically able} only builds on that first debate (see \cite{duijf2018let}, \cite{naumov2017together}, \cite{naumov2017coalition}, \cite{aagotnes2015knowledge}, \cite{wang2015logic}, \cite{xstit}, \cite{HortyPacuit2017}). Presently, we base the conceptual reach of know-how on the approach of epistemic stit. Much like Horty \& Pacuit in \cite{HortyPacuit2017},  we are concerned with situations in which agentive knowledge yields differences between what agents can bring about, on one hand, and what they can bring about \emph{knowingly}, on the other. Consider Example \ref{ej1} and the differences between case a) and case b). In case a), we would like to say that $luther$ and $benji$  do not know how to save the facility. Whether being \emph{causally} able to perform an action that they know will save the facility is equal to knowing how to save it or not is very much open to debate, but the reader will agree that at least it is necessary for knowing how. Broersen, Herzig \& Troquard, Horty \& Pacuit, Naumov \& Tao, \AA gotnes et al., and Duijf all agree with this, and \cite{herzig2006knowing} and \cite{HortyPacuit2017} in fact characterize their versions of know-how exactly as the possibility of bringing about a certain outcome knowingly. In what follows, we focus our study of previous literature on \emph{individual} know-how, leaving its collective counterpart for future endeavors.  

 Horty \& Pacuit (\cite{HortyPacuit2017} and \cite{hortyepistemic}) work on the assumption that an agent is epistemically able to do something iff it is able to have \emph{ex interim} knowledge of that something. Their goal is to formally disambiguate simplified versions of the different cases of our Example \ref{ej1}. Using a logic endowed with the operator $[\dots kstit]$ that we reviewed before, they define know-how as the historical possibility of \emph{ex interim} knowledge. We have already commented on the problems of their system. 
 
 Another related approach to know-how is given by Herzig \& Troquard in the already mentioned \cite{herzig2006knowing}.  Herzig \& Troquard use their dynamic knowledge operator $[\dots Kstit]$ to state that an agent knows how to see to it that $\phi$ (in the next moment) iff it is historically necessary that the agent knowingly enforces the possibility to knowingly bring about $\phi$: at situation $\langle m, h\rangle$, $\alpha$ knows how to see to it that $\phi$ iff $\langle m, h \rangle\models \square[\alpha\ Kstit]\Diamond [\alpha\ Kstit]X\phi$. In our view, Herzig \& Troquard's treatment of know-how is successful to a certain extent, but their models are constrained by $(\mathtt{OAC})$ and by the condition of uniformity corresponding to Horty \& Pacuit's  $(\mathtt{UAAT})$, as we established before. Somewhat connected to Herzig \& Troquard's version, Duijf presents in \cite{duijf2018let} a simple and elegant stit theory of know-how that ultimately characterizes individual know-how by stating that at situation $\langle m,h \rangle$, $\alpha$ knows how to see to it that $\phi$ iff $\langle m, h \rangle\models \mathds{K}_\alpha \Diamond \mathds{K}_\alpha [\alpha]\phi$ --where $\mathds{K}_\alpha$ are the operators for Duijf's version of \emph{ex interim} knowledge and $[\alpha]$ are the traditional Chellas-stit operators. As we saw, Duijf rejects imposing condition ($\mathtt{OAC}$) on his frames, so his logic admits finer-grained renditions of knowledge. However, just as in Herzig \& Troquard's proposal, the condition of uniformity of available actions --which in Duijf's system is syntactically expressed by the axiom schema $\Diamond \mathds{K}_\alpha[\alpha] \phi\to \mathds{K}_\alpha\Diamond[\alpha]\phi$--  yields that his formula for know-how is equivalent to $\Diamond \mathds{K}_\alpha [\alpha]\phi$.\footnote{Duijf does not comment on such equivalence, whose deduction --in Duijf's system-- comes from the following argument. Turns out to be the case that the validity of $\Diamond \mathds{K}_\alpha[\alpha] \phi\to \mathds{K}_\alpha\Diamond[\alpha]\phi$ entails the validity of $\Diamond \mathds{K}_\alpha\phi\to \mathds{K}_\alpha\Diamond\phi$ in Duijf's logic. This last schema, denoted by $(Unif-H)$, is the syntactic counterpart of a condition that we call `uniformity of historical possibility', and in light of it we have that   \[ \begin{array}{lll}
1. \vdash \mathds{K}_\alpha[\alpha]\phi \to \mathds{K}_\alpha \mathds{K}_\alpha[\alpha]\phi && \mbox{Substitution of axiom }(4) \mbox{ for } \mathds{K}_\alpha\\
2. \vdash \Diamond \mathds{K}_\alpha [\alpha]\phi \to \Diamond \mathds{K}_\alpha  \mathds{K}_\alpha[\alpha]\phi && \mbox{Modal logic on } 1\\
3. \vdash \Diamond \mathds{K}_\alpha  \mathds{K}_\alpha[\alpha]\phi\to \mathds{K}_\alpha \Diamond \mathds{K}_\alpha[\alpha]\phi && \mbox{Substitution of }(Unif-H) \mbox{ for } \alpha\\
4.  \vdash \Diamond \mathds{K}_\alpha [\alpha]\phi\to \mathds{K}_\alpha \Diamond \mathds{K}_\alpha[\alpha]\phi && \mbox{Propositional logic } 2,4.
\end{array} \] The other direction is straightforward, by axiom $(T)$ for $\mathds{K}_\alpha$. A similar deduction can be provided to ensure that Herzig \& Troquard's $\square [\alpha\ Kstit]\Diamond [\alpha\ Kstit]\phi$ is reducible to $\Diamond [\alpha\ Kstit]\phi$.}

\textbf{Our version:} Let $\mathcal{M}$ be an \emph{epistemic branching discrete-time frame} and $\phi$ of $\mathcal L_{\textsf{KX}}$. We say that at situation $\langle m,h \rangle$, $\alpha$ knows how to see to it that $\phi$ iff $\langle m, h \rangle\models \square K_\alpha \Diamond K_\alpha [\alpha]^X\phi$. So we propose that an agent  knows how to see to it that $\phi$ if it has \emph{ex ante} knowledge of the possibility of knowing \emph{ex interim} (or knowingly doing) $\phi.$\footnote{We observe that if we were to incorporate a condition of uniformity of available actions into our logic (as we do in the axiomatization), it would be equivalent to the semantic condition known as `uniformity of historical possibility' $(\mathtt{Unif-H})$, which says that for every situation $\langle m_*,h_*\rangle$, if $\langle m_*, h_*\rangle \sim_\alpha \langle m, h\rangle $ for some $\langle m,h\rangle $, then for every $h_*'\in H_{m_*}$ there exists $h'\in H_m$ such that $\langle m_*, h_*'\rangle \sim_\alpha \langle m, h'\rangle$. Under this condition, which corresponds syntactically to the schema $\Diamond K_\alpha\phi\to K_\alpha\Diamond\phi$, we would have two important consequences: our formula for \emph{ex ante} knowledge would be equivalent to $\square K_\alpha X\phi$, and our formula for know-how would be equivalent to $\Diamond K_\alpha[\alpha]^X \phi$.} In this way, in Example \ref{ej1} a) we have that $\mathcal{M},\langle m_4, h_{10}\rangle \not \models \square K_{luther} \Diamond K_{luther}[luther]^X s$ ($luther$ does not know how to defuse the bombs), whereas in case b) we have that $\mathcal{M},\langle m_4, h_{10}\rangle \models \square K_{luther} \Diamond K_{luther}[luther]^X d_L$ ($luther$ knows how to set off a bomb). Since in both cases, $luther$ ultimately does set off a bomb, we consider case a) as a situation where he should be excused from moral responsibility of the explosion, while case b) is one where he should be held morally responsible for it.\footnote{Although we focus our comparisons on the previous work within epistemic stit, it is worth discussing some of the approaches to the concept in the epistemic extensions of ATL and Coalition Logic (see \cite{aagotnes2015knowledge}, \cite{HoekWooldridgeATEL2003}), for the ideas behind the syntax and semantics for know-how in these logics are similar to those of stit. For instance, Naumov \& Tao and \AA gotnes et al. --in \cite{naumov2017together} and  \cite{aagotnes2015knowledge}, respectively-- share many intuitions with Horty \& Pacuit. The notion of know-how they both formalize is characterized by the statement that an agent knows how to bring about $\phi$ at a given state $s$ iff there exists a `strategy' $a$ such that in all states that are epistemically indistinguishable to $s$ for the agent, `strategy' $a$ will lead to states at which $\phi$ holds. In other words, an agent knows how to do something if there exists a way for the agent to knowingly enforce $\phi$. \cite{naumov2017together} and \cite{aagotnes2015knowledge} use different interpretations for the word `strategy'. While in the former the authors refer to an action label in single-step transitions, \AA gotnes et al. use the term as is done in ATL, where strategies are functions that assign to each agent and state a pertinent transition. Regardless of the difference, their formalization of know-how depends on the same reasoning: an agent would know how to do $\phi$ iff there exists a uniform strategy such that at all epistemically indistinguishable states, the transition assigned by the strategy leads to a state at which $\phi$ holds. In both accounts, we face again the idea of uniformity.}

\section{Axiomatization}\label{stormbreaker}

\begin{definition}[Proof system]
\label{axiomsystemurakami}
Let $\Lambda$ be the proof system defined by the following axioms and rules of inference: \begin{itemize}
    \item \emph{(Axioms)} All classical tautologies from propositional logic. The $\mathbf{S5}$ axiom schemata for $\square$, $[\alpha ]$, $[Ags]$, and $K_\alpha$. The following axiom schemata for the interactions of formulas with the given operators:
\scriptsize
\begin{flalign*}
&YX\phi \leftrightarrow \phi &(In1)\\ 
&XY\phi \leftrightarrow \phi &(In2)\\ 
&X\phi \leftrightarrow \lnot X \lnot\phi &(DET.S.X)\\ 
&Y\phi \leftrightarrow \lnot Y \lnot\phi &(DET.S.Y)\\ 
&\square\phi \to[\alpha]\phi &(SET)\\ 
&[\alpha]X \phi\to [\alpha]X\square \phi&(NA)\\ 
&[Ags]X \phi\to [Ags]X\square \phi& (NAgs)\\ 
&[\alpha] \phi\to [Ags] \phi& (GA)\\ 
&\mbox{For $m\geq 1$ and pairwise distinct $\alpha_1,\dots,\alpha_m$},& \\ &\bigwedge_{1\leq i\leq n}  \Diamond [\alpha_i ] p_i \to \Diamond\left(\bigwedge_{1\leq i\leq n}[\alpha_i ] p_i\right)& (IA)\\
&K_\alpha X\phi \to X K_\alpha \phi &(NoF)\\
&\Diamond K_\alpha p \to K_\alpha \Diamond  p&(Unif-H)\\
\end{flalign*}
\normalsize
\item \emph{(Rules of inference)} \textit{Modus Ponens}, Substitution, and Necessitation for the modal operators.
\end{itemize}
We define $\Lambda_n$ as the axiom system constructed by adding axiom $(AgsPC_n)$ to $\Lambda$, where
\scriptsize
\begin{flalign*}
&\bigwedge_{1\leq k\leq n}\Diamond \left(\left(\bigwedge_{1\leq i \leq k-1}\lnot \phi_i\right) \land [Ags] \phi_k\right)\to \bigvee_{1\leq k\leq n} \phi_k &(AgsPC_n).
\end{flalign*}
\end{definition}
Following \cite{broersen2008complete} and \cite{payette2014decidability}, we will show that the axiom system $\Lambda_n$ is sound and complete with respect to the general multi-modal Kripke models, which Payette calls `irregular' in \cite{payette2014decidability}. However, we conjecture that there is no problem in using unraveling techniques as in \cite{schwarzentruber2012complexity} to transform these general Kripke models into those of Definition \ref{frames} in a truth-preserving way. 

\begin{proposition}\label{soundncomp} The system $\Lambda_n$ is sound and complete with respect to the class of Kripke-exstit $n$-models. 
\end{proposition}

The proof of soundness is straightforward. For completeness, we proceed in three steps. In the first step, we prove completeness with respect to Kripke models that are \emph{super-additive}, meaning those for which each action available to $Ags$ is included in an intersection of individual actions but is not necessarily the same as such intersection (see \cite{schwarzentruber2012complexity} for their exact definition). In the second step, we prove completeness with respect to Kripke super-additive models where the temporal relations are irreflexive. In the last step, we use a technique similar to Schwarzentruber's construction in \cite{schwarzentruber2012complexity} and Lorini's in \cite{lorini2013temporal} to prove completeness with respect to the class of actual models, meaning those for which each action available to $Ags$ is the same as an intersection of individual actions. The full proof is long and technical and can be found in the Appendix \ref{slumber} of the present work.

\section{Conclusion}
In this work, we carefully reviewed previous renditions of four kinds of agentive knowledge (\emph{ex ante}, \emph{ex interim}, \emph{ex post,} and know-how) in epistemic stit theory. Motivated by examples that demand a notion of flexibility of the epistemic component in analyses of responsibility attribution, we presented a new logic for them. We find that our versions offer a fine background for building a nuanced theory of responsibility based on the influence of knowledge on action and decision.


\bibliographystyle{splncs04}

\bibliography{AAAI2018references}
\clearpage

\begin{subappendices}

\renewcommand{\thesection}{\Alph{section}}%

\section{Soundness and Completeness of $\Lambda_n$} \label{slumber}

\begin{remark}
As some remarks for the axiom system presented in Section \ref{stormbreaker}, we have that \begin{itemize}
    \item $(SET)$ and $(GA)$ imply that $\vdash \square \phi\to [Ags]\phi$, so that $(NAgs)$ and $(T)$ for $[Ags]$ imply that $\vdash\square X \phi\to X\square \phi$, which is a theorem that we will refer to by $(NX)$.
    \item Necessitation of $Y$ and $\square$, together with $(In2)$, entails that $(NX)$ implies that $\vdash Y\square  \phi\to \square Y \phi$, a theorem which we call $(NY)$. 
    \item Schema $(GA)$ and axiom $(AgsPC_n)$ imply that
    \scriptsize
    \[\vdash\bigwedge_{1\leq k\leq n}\Diamond \left(\left(\bigwedge_{1\leq i \leq k-1}\lnot \phi_i\right) \land [\alpha] \phi_k\right)\to \bigvee_{1\leq k\leq n} \phi_k,\] \normalsize which is a theorem that we call $(APC)_n$. This theorem will encode the fact that at each moment, each agent will have at most $n$ choices of action to decide from. 
    \item Necessitation of $Y$ and $K_\alpha$,
    together with $(In2)$, entails that $(NoF)$ implies that $\vdash YK_\alpha  \phi\to K_\alpha Y \phi$. 
\end{itemize}  
    
\end{remark}

\begin{definition}[Kripke-exstit frames]\label{ledbetter}

A tuple $\mathcal{F}=\langle W, R_\square, R_X, R_Y ,\mathtt{Choice}, \{\mathtt{\approx}_\alpha\}_{\alpha\in Ags}\rangle$ is called a Kripke-exstit frame iff $W$ is a set of possible worlds and

\begin{enumerate}

\item $R_\square$ is an equivalence relation over $W$. For $w\in W$, the class of $w$ under $R_\square$ is denoted by $\overline{w}$.

\item $R_X$ and $R_Y$ are serial and deterministic relations on $W$ that fulfill the following conditions \begin{itemize}
    \item $(\mathtt{Inverse})$ $R_X\circ R_Y=Id$ and $R_Y\circ R_X=Id$.
    \item $(\mathtt{NX})_K$  $R_\square\circ R_X \subseteq R_X \circ R_\square$.\footnote{Observe that with $(\mathtt{Inverse})$ this condition implies $(\mathtt{NY})_K$ $:  R_Y \circ R_\square \subseteq R_\square\circ R_Y$. }
\end{itemize} 

\item $\mathtt{Choice}$ is a function satisfying the following properties: 
\begin{itemize}
    \item It assigns to each $\alpha\in Ags$ 
    a partition $\mathtt{Choice}_\alpha$ of $W$ given by an equivalence relation which we will denote by $R_\alpha$. For $w\in W$, the class of $w$ in the partition $\mathtt{Choice}_\alpha$ is denoted by $\mathtt{Choice}_\alpha(v)$.
    \item It assigns to the grand coalition $Ags$ a partition $\mathtt{Choice}_{Ags}$ of $W$ such that $\mathtt{Choice}_{Ags}(v)=\bigcap\limits_{\alpha \in Ags} \mathtt{Choice}_\alpha(v)$ for each $v\in W$. We denote the equivalence relation defining such partition $R_{Ags}$. 
    
\end{itemize}
Moreover, $\mathtt{Choice}$ must satisfy the following constraints:
\begin{itemize}
\item $(\mathtt{SET})_K$ For every $w \in W$, we have that $\mathtt{Choice}_\alpha(w)\subseteq \overline{w}$ for every $\alpha\in Ags$. This implies that the set $\{ \mathtt{Choice}_\alpha (v);v\in \overline{w}\}$ is a partition of $\overline{w}$ for every $\alpha\in Ags$, which we will denote by $\mathtt{Choice}_\alpha^{\overline{w}}$. Similarly, it implies that $\mathtt{Choice}_{Ags}(w)\subseteq \overline{w}$ and that the set $\{ \mathtt{Choice}_{Ags} (v);v\in \overline{w}\}$ is a partition of $\overline{w}$, which we will denote by $\mathtt{Choice}_{Ags}^{\overline{w}}$. 
\item $\mathtt{(IA)_K}$ For $w\in W$, we have that each function $s:Ags\to \mathcal{P}(\overline{w})$ that maps $\alpha$ to a member of $\mathtt{Choice}_\alpha^{\overline{w}}$ is such that $\bigcap_{\alpha \in Ags} s(\alpha) \neq \emptyset$. 
\item $\mathtt{(NA)_K}$ For $\alpha \in Ags$ and $w\in W^\Lambda$, $R_\square\circ R_X\circ R_{\alpha}\subseteq R_X\circ R_{\alpha}.$
    \item $\mathtt{(NAgs)_K}$ For $w\in W^\Lambda$, $R_\square\circ R_X\circ R_{Ags}\subseteq R_X\circ R_{Ags}.$

\end{itemize}

\item For each $\alpha \in Ags$, $\approx_\alpha$ is an (epistemic) equivalence relation on $W$ that satisfies the following constraints: 
\begin{itemize}
\item $\mathtt{(Unif-H)_K}$ Let $w_1, w_2\in W$ such that there exist $v\in \overline{w_1}$ and $u\in \overline{w_2}$ with $v\approx_\alpha u$. Then for every $v'\in \overline{w_1}$, there exists $u'\in \overline{w_2}$ such that $v'\approx_\alpha u'$.
\item $\mathtt{(NoF)_K}$ For $\alpha\in Ags$, we have that $\approx_\alpha\circ R_X\subseteq R_X\circ\approx_\alpha$. 

\end{itemize}
\end{enumerate}
\end{definition}

\begin{remark}

Frames for which the group-action condition in item 3 of the definition above is relaxed to $\mathtt{Choice}_{Ags}(v)\subseteq\bigcap\limits_{\alpha \in Ags} \mathtt{Choice}_\alpha(v)$ for each $w\in W$ and $v\in \overline{w}$ are called \emph{\textbf{super-additive}}.

Frames where for every $w\in W$ the cardinalities of the partitions $\mathtt{Choice}_{Ags}^{\overline{w}}$ and $\mathtt{Choice}_{\alpha}^{\overline{w}}$ (for every $\alpha\in Ags$) are at most $n$ will be called $n$-frames. 

\end{remark}

\begin{definition}

A Kripke-exstit model $\mathcal {M}$ consists of the tuple that results from adding a valuation function $\mathcal{V}$ to a Kripke-exstit frame, where $ \mathcal{V}: P\to 2^{T \times H}$ assigns to each atomic proposition a set of worlds. The semantics for the formulas of $\mathcal {L}_{\textsf{KX}}$ is defined recursively by the pertinent truth conditions, mirroring Definition \ref{models KCSTIT}. If we add a valuation like this to a tuple defining a super-additive frame, we will also refer to the model as super-additive. If we add a valuation like this to a tuple defining an $n$-frame, we refer to the model as an $n$-model.  

\end{definition}

\begin{proposition}[Soundness]\label{soundness} The proof system $\Lambda_n$ is sound with respect to the class of Kripke-exstit $n$-models.
\end{proposition}

It is clear that each axiom corresponds to the appropriate relational property in the definition of Kripke-exstit $n$-models. Therefore, the proof of soundness is routine.


\subsection{Canonical Kripke-exstit models}

We will show that the proof system $\Lambda_n$ is complete with respect to the class of super-additive Kripke-exstit $n$-models.

The strategy is to build a canonical structure from the syntax.  
\begin{definition}[Canonical Structure]
\label{dicaprio} For $n\geq 1$, 
the tuple \[\mathcal{M}=\langle W^\Lambda, R_\square, R_X, R_Y, \mathtt{Choice}, \{\mathtt{\approx}_\alpha\}_{\alpha\in Ags},\mathcal{V} 
\rangle \] is called a canonical structure for $\Lambda_n$ iff

\begin{itemize}
\item $W^\Lambda=\{w ;w  \mbox{ is a } \Lambda\textnormal{-MCS}\}$. 

\item $R_\square$ is a relation on $W^\Lambda$ defined by the following rule:  for $w,v\in W^\Lambda$, $wR_{\square}v$ iff for every $\phi$, $\square\phi\in w\Rightarrow \phi\in v$. For $w\in W^\Lambda$, the set $\{v\in W^\Lambda ;wR_\square v \}$ is denoted by $\overline{w}$. 

\item $R_X$ is a relation on $W^\Lambda$ defined by the following rule:  for $w,v\in W^\Lambda$, $wR_{X}v$ iff for every $\phi$, $X\phi\in w\Rightarrow \phi\in v$. For $w\in W^\Lambda$, the set $\{v\in W^\Lambda ;wR_X v \}$ is denoted by $x[w]$. 

\item $R_Y$ is a relation on $W^\Lambda$ defined by the following rule:  for $w,v\in W^\Lambda$, $wR_{Y}v$ iff for every $\phi$, $Y\phi\in w\Rightarrow \phi\in v$. For $w\in W^\Lambda$, the set $\{v\in W^\Lambda ;wR_Y v \}$ is denoted by $y[w]$. 
\item $\mathtt{Choice}$ is a function that fulfills the following requirements: \begin{itemize}
    \item It assigns to each $\alpha$ a subset of $\mathcal{P}(W^\Lambda)$, which will be denoted by $\mathtt{Choice}_\alpha$ and defined as follows: let $R_\alpha$ be a relation on $W^\Lambda$ such that for $w,v\in W^\Lambda$, $wR_{\alpha}v$ iff for every $\phi$, $[\alpha ]\phi\in w\Rightarrow \phi\in v$; if we take $\mathtt{Choice}_\alpha(v)=\{u\in W^\Lambda ; vR_\alpha u\}$, then we set $\mathtt{Choice}_\alpha=\bigcup_{v\in W^\Lambda}\mathtt{Choice}_\alpha(v)$.
    \item It assigns to the grand coalition $Ags$ a subset of $\mathcal{P}(W^\Lambda)$ denoted by $\mathtt{Choice}_{Ags}$, and defined as follows: let $R_{Ags}$ be a relation on $W^\Lambda$ such that for $w,v\in W^\Lambda$, $wR_{Ags}v$ iff for every $\phi$, $[Ags ]\phi\in w\Rightarrow \phi\in v$;  we set $\mathtt{Choice}_{Ags}$ in an analogous way to the $\mathtt{Choice}_\alpha$.
\end{itemize} 

\item For each $\alpha \in Ags$, $\approx_\alpha$ is an epistemic relation on $W^\Lambda$ given by the following rule:  for $w,v\in W^\Lambda$, $w\approx_{\alpha}v$ iff for every $\phi$, $K_\alpha\phi\in w\Rightarrow \phi\in v$. 


\item $\mathcal{V}$ is the canonical valuation, defined such that $w\in \mathcal{V}(p)$ iff $p\in w$. 

\end{itemize}
\end{definition}

\begin{proposition}\label{can}
For $n\geq 1$, the canonical structure $\mathcal{M}$ for $\Lambda_n$ is a Kripke-exstit super-additive model where the cardinalities of the partitions $\mathtt{Choice}_{Ags}$ and $\mathtt{Choice}_{\alpha}$ (for every $\alpha\in Ags$) are at most $n$ --therefore an $n$-model. 
\end{proposition}

\begin{proof}
We want to show that $\langle W^\Lambda, R_\square, R_X, R_Y, \mathtt{Choice},\{\mathtt{\approx}_\alpha\}_{\alpha\in Ags}\rangle$ is a Kripke-exstit super-additive  frame, which amounts to showing that the tuple validates the four items in the definition of Kripke-exstit super-additive  frames. 

\begin{enumerate}

\item It is clear that $R_\square$ is an equivalence relation, since $\Lambda$ includes the $\mathbf{S5}$ axioms for $\square$. 

\item Seriality and determinicity of $R_X$ and $R_Y$ come from the axioms $(DET.S.X)$ and $(DET.S.Y)$, respectively, by the following arguments. For seriality, let $v\in W^\Lambda$. We will show that   
$z':=\{\psi ; Y\psi\in v \}$ and $x':=\{\psi ; X\psi\in v \}$ are consistent. We show it for $z'$ and assume and analogous argument for $x'$.
Suppose for a contradiction that $z'$ is not consistent. Then there exists a set $\{\psi_1,\dots,\psi_m\}$ of formulas of $\mathcal{L}_{\textsf{KX}}$ such that $\{\psi_1,\dots,\psi_m\}\subseteq \{\psi ; Y\psi\in v \}$ and $\vdash \psi_1\wedge\dots\wedge \psi_m\to \bot$ (a). Now, the fact that $\{\psi_1,\dots,\psi_m\}\subseteq \{\psi ; Y\psi\in v \}$ means that $Y\psi_i\in v'$ for every $1\leq i \leq m$; by necessitation of $Y$ and its distributivity over conjunction, we get that (a) implies that $\vdash Y\psi_1\wedge\dots\wedge Y\psi_m\to Y\bot$, which by $(DET.S.Y)$ implies that $\vdash Y\psi_1\wedge\dots\wedge Y\psi_m\to \langle Y\rangle\bot$, but this is a contradiction, since $
v$ is a $\Lambda$-MCS which includes $Y\psi_1\wedge\dots\wedge Y\psi_m$. Let $z$ be the $\Lambda$-MCS that includes $z'$, and $x$ be the $\Lambda$-MCS that includes $x'$. It is clear that $vR_Y z$ and that $vR_X x$. For determinicity, suppose that, besides the existent $z$ and $x$ we showed above, there exist $z_*$ and $x_*$ such that $vR_Y z_*$ and $vR_X x_*$. We show that $z_*=z$ and $x_*=x$. We show it for $z_*$ and assume an analogous argument for $x_*$. For any $\phi$ of $\mathcal{L}_{\textsf{KX}}$, $\phi\in z_*$ iff $\langle Y \rangle \phi\in v$ iff (using $(DET.S.Y)$) $Y\phi \in v$ iff $\phi \in z$. Therefore, $z_*=z$.

Axioms $(In1)$ and $(In2)$ ensure that $R_X\circ R_Y=Id$ and that $R_Y\circ R_X=Id$, respectively, by the following arguments. We show that $R_X\circ R_Y=Id$ and assume an analogous argument for $R_Y\circ R_X$. Let $v\in W^{\Lambda}$, and let $z$ be the unique $\Lambda$-MCS such that $vR_Y z$. We show that $zR_X v$: assume that $X\phi \in z$. Axiom $(DET.S.X)$ and the fact that $z$ is a $\Lambda$-MCS implies that $X\lnot\phi\not\in z$. Suppose for a contradiction that $\phi\not\in v$. Since $v$ is a $\Lambda$-MCS, $\lnot\phi\in v$, so that axiom $(In1)$ implies that $YX\lnot\phi\in v$. By definition of $z$, this implies that $X\lnot\phi\in z$, which is a contradiction. Therefore $\phi\in v$ and thus $zR_X v$. 

In turn, $(NX)$ guarantees that $R_\square\circ R_X \subseteq R_X \circ R_\square$ by the following arguments.
Assume that $wR_\square\circ R_X v$ via $y$.  Let $z$ be the unique $\Lambda$-MCS such that $vR_Y z$. We have seen that $zR_X v$. Now we show that $wR_\square z$: assume that $\square\phi\in w$. Axiom $(In2)$, axiom $(K)$ for $\square$, and necessitation for $\square$ imply that $\square XY\phi\in w$. Axiom $(NX)$ then implies that $X\square Y\phi\in w$. By construction, this implies that $\square Y\phi\in y$, which in turn implies that $Y\phi\in v$ and thus --by definition-- that $\phi\in z$. Therefore, we have that $z$ is such that $wR_\square z$ and $zR_X v$, which means that $w R_X\circ R_\square v$. Therefore, we have that $R_\square\circ R_X \subseteq R_X \circ R_\square$. 

\item Since $\Lambda$ includes the $\mathbf{S5}$ axioms for $[\alpha ]$ (for each $\alpha\in Ags$), we have that $R_\alpha$ is an equivalence relation for each $\alpha\in Ags$. 

We have to show that the condition of \emph{\textbf{super-additivity}} is satisfied, which amounts to showing that $R_{Ags}\subseteq \bigcap_{\alpha \in Ags} R_{\alpha}$. It is clear that schema $(GA)$ entails precisely this: suppose that $wR_{Ags} v$, and assume that $[\alpha]\phi\in w$. Because of $(GA)$, this last thing implies that $[Ags]\phi\in w$ as well, so that the supposition that $wR_{Ags} v$  yields that $\phi \in v$. Therefore, we have that $w R_{\alpha}v$, but since $\alpha$ was taken arbitrarily, we have shown that $R_{Ags}\subseteq \bigcap_{\alpha \in Ags} R_{\alpha}$.

We must verify that $\mathcal{M}$ validates the constraints $(\mathtt{SET})_K$ $(\mathtt{IA})_K$, $(\mathtt{NA})_K$, and $(\mathtt{NAgs})_K$: \begin{itemize}
\item $(\mathtt{SET})_K$ Since $\Lambda$ includes $\square \phi\to [\alpha ]\phi$ as an axiom schema, we have that for every $w\in W^\Lambda$, $\mathtt{Choice}_\alpha (w)\subseteq \overline{w}$ for every $\alpha \in Ags$. 
    \item $(\mathtt{IA})_K$ In order to show it, we need two intermediate results: \begin{enumerate}[a)]
\item For a fixed $w_*\in W^\Lambda$, we have that $w\in\overline{w_*}$ iff $\{\square \psi; \square\psi\in w_*\}\subseteq w$. 

$(\Rightarrow)$ Let $w\in\overline{w_*}$ (which means that $w_*R_\square w$). Take $\phi$ of $\mathcal{L}_{\textsf{Kx}}$ such that $\square\phi\in w_*$. Since $w_*$ is closed under \emph{Modus Ponens}, axiom $(4)$ for $\square$ implies that $\square\square\phi\in w_*$ as well. Therefore, by the definition of $R_\square$, we get that $\square\phi \in w$. 

$(\Leftarrow)$ We assume that $\{\square \psi; \square\psi\in w_*\}\subseteq w$. Take $\phi$ of $\mathcal{L}_{\textsf{KX}}$ such that $\square\phi\in w_*$. By our assumption, we get that $\square \phi\in W$. Since $w$ is closed under \emph{Modus Ponens}, axiom $(T)$ for $\square$ implies that $\phi\in w$ as well. In this way, we have that the fact that $\square \phi\in w_*$ implies that $\phi\in w$, which means that $w_*R_\square w$ and $w\in\overline{w_*}$. 

\item For a fixed $w_*\in W^\Lambda$ and $s:Ags\to \mathcal{P}(\overline{w_*})$ a function that maps $\alpha$ to a member of $\mathtt{Choice}^{\overline{w_*}}_\alpha$ such that $v_{s(\alpha)}\in s(\alpha)$, we have that $w\in s(\alpha)$ iff $\Delta_{s(\alpha)}=\{[\alpha ] \psi ;[\alpha ]\psi\in v_{s(\alpha)}\}\subseteq w$.

$(\Rightarrow)$ Let $w\in s(\alpha)$ (which means that $v_{s(\alpha)}R_\alpha w$). Take $\phi$ of $\mathcal{L}_{\textsf{KX}}$ such that $[\alpha ]\phi\in v_{s(\alpha)} $. Since $v_{s(\alpha)}$ is closed under \emph{Modus Ponens}, schema $(4)$ for $[\alpha ]$ implies that $[\alpha ][\alpha ]\phi\in v_{s(\alpha)}$ as well. Therefore, by definition of $R_\alpha$, we get that $[\alpha ]\phi \in w$. 

$(\Leftarrow)$ We assume that $\Delta_{s(\alpha)}=\{[\alpha ] \psi ; [\alpha ]\psi\in v_{s(\alpha)}\}\subseteq w$. Take $\phi$ of $\mathcal{L}_{\textsf{KX}}$ such that $[\alpha ]\phi\in v_{s(\alpha)}$. By our assumption, we get that $[\alpha ]\phi\in w$. Since $w$ is closed under \emph{Modus Ponens}, the $(T)$ axiom for $[\alpha ]$ implies that $\phi\in w$ as well. In this way, we have that the fact that $[\alpha ] \phi\in v_{s(\alpha)}$ implies that $\phi\in w$, which means that $v_{s(\alpha)}R_\alpha^{\overline{w_*}} w$ and $w\in s(\alpha)$.
\end{enumerate}

Next, we will show that for a fixed $w_*\in W^\Lambda$ and $s:Ags\to \mathcal{P}(\overline{w_*})$ a function that maps $\alpha$ to a member of $\mathtt{Choice}^{\overline{w_*}}_\alpha$ such as in item b) above, we have that $\bigcup_{\alpha \in Ags}\Delta_{s(\alpha)}\cup \{\square \psi ; \square\psi\in w_*\}$ is $\Lambda$-consistent.

First, we will show that $\bigcup_{\alpha\in Ags} \Delta_{s(\alpha)}$ is consistent. Suppose that this is not the case. Then there exists a set $\{\phi_1,\dots,\phi_n\}$ of formulas of $\mathcal{L}_{\textsf{KX}}$ such that $[\alpha_{i} ]\phi_i\in v_{s(\alpha_i)}$ for every $1\leq i \leq n$ and \begin{equation}\label{paul}
\vdash([\alpha_{1} ]\phi_1\wedge\dots\wedge [\alpha_{n}  ]\phi_n)\to \bot.
\end{equation}

Without loss of generality, we assume that $\alpha_i\neq\alpha_j$ for all $j\neq i$ such that $j,i\in \{1,\dots,n\}$ --this assumption hinges on the fact that any stit operator distributes over conjunction. Notice that the fact that $[\alpha_{i} ]\phi_i\in v_{s(\alpha_i)}$ for every $1\leq i \leq n$ implies that $\Diamond[\alpha_{i} ]\phi_i \in w_*$ for every $1\leq i \leq n$. 
Since $w_*$ is closed under conjunction, we also have that $\Diamond [\alpha_{1} ]\phi_1\wedge\dots\wedge\Diamond[\alpha_{n} ]\phi_n\in w_*$.

By axiom $(IA)$, we have that \begin{equation}\label{paul2} 
\vdash\Diamond [\alpha_{1} ]\phi_1\wedge\dots\wedge\Diamond[\alpha_{n} ]\phi_n\to \Diamond( [\alpha_{1} ]\phi_1\wedge\dots\wedge[\alpha_{n} ]\phi_n).
\end{equation} Therefore, equations \eqref{paul2} and \eqref{paul},  imply that \begin{equation}\label{paul3} 
\vdash\Diamond [\alpha_{1} ]\phi_1\wedge\dots\wedge\Diamond[\alpha_{n} ]\phi_n\to \Diamond \bot.
\end{equation}
But this is a contradiction, since we had seen that $\Diamond [\alpha_{1} ]\phi_1\wedge\dots\wedge\Diamond[\alpha_{n} ]\phi_n\in w_*$, and $w_*$ is a $\Lambda$-MCS. Therefore, $\bigcup_{\alpha\in Ags} \Delta_{s(\alpha)}$ is consistent.

Next, we show that the union $\bigcup_{\alpha\in Ags} \Delta_{s(\alpha)}\cup  \{\square\psi ; \square\psi \in w_*\}$ is also consistent. Suppose that this is not the case. Since $\bigcup_{\alpha\in Ags} \Delta_{s(\alpha)}$ is consistent, there must exist sets $\{\phi_1,\dots,\phi_n\}$ and $\{\theta_1,\dots,\theta_m\}$ of formulas of $\mathcal{L}_{\textsf{KX}}$ such that $[\alpha_i ]\phi_i\in v_{s(\alpha_i)}$ for every $1\leq i \leq n$, $\square\theta_i\in w_*$ for every $1\leq i\leq m$, and \begin{equation}\label{newman} \vdash([\alpha_1 ]\phi_1\wedge\dots\wedge [\alpha_n]\phi_n)\wedge (\square\theta_1\wedge\dots\wedge\square\theta_m)\to\bot. 
\end{equation} Let $\theta=\theta_1\land\dots \land\theta_m$. Since $\square$ distributes over conjunction, we have that  $\vdash\square\theta\leftrightarrow\square\theta_1\wedge\dots\wedge\square\theta_m$, where it is important to mention that since $w_*$ is a $\Lambda$-MCS, then $\square\theta\in w_*$. In these terms, \eqref{newman} implies that \begin{equation}\label{newman2}\vdash([\alpha_1 ]\phi_1\wedge\dots\wedge [\alpha_n ]\phi_n)\to\lnot \square \theta.\end{equation} Once again, we assume without loss of generality that $\alpha_i\neq\alpha_j$ for all $j\neq i$ such that $j,i\in \{1,\dots,n\}$. Analogous to the procedure we used to show that $\bigcup_{\alpha\in Ags} \Delta_{s(\alpha)}$ is consistent, \eqref{newman2} implies that \begin{equation}\label{newman3} 
\vdash\Diamond [\alpha_{1} ]\phi_1\wedge\dots\wedge\Diamond[\alpha_{n} ]\phi_n\to \Diamond \lnot \square\theta.
\end{equation}

This entails that $\Diamond \lnot \square\theta\in w_*$, but this is a contradiction, since the fact that $\square\theta\in w_*$ implies with axiom $(4)$ for $\square$ that $\square\square\theta\in w_*$. 

Now, let $u_*$ be the $\Lambda$-MCS that includes $\bigcup_{\alpha\in Ags} \Delta_{s(\alpha)}\cup  \{\square\psi ; \square\psi \in w_*\}$. By the intermediate result  a), it is clear that $u_*\in \overline{w_*}$. By the intermediate result b), it is clear that $u_*\in s(\alpha)$ for every $\alpha\in Ags$. Therefore, we have shown that for a fixed $w_*\in W$, we have that each function $s:Ags\to \mathcal{P}(\overline{w_*})$ that maps $\alpha$ to a member of $\mathtt{Choice}^{\overline{w_*}}_\alpha$ is such that $\bigcap_{\alpha \in Ags} s(\alpha) \neq \emptyset$, which means that $\mathcal{M}$ validates the constraint $(\mathtt{IA})_K$.

\item $\mathtt{(NA)_K}$ We want to show that for $\alpha \in Ags$, $R_\square\circ R_X\circ R_{\alpha}\subseteq R_X\circ R_{\alpha}.$ Let $\alpha \in Ags$ and $v, o\in W^\Lambda$ such that $v R_\square\circ R_X\circ R_{\alpha} o$, which means that there exist $v',o'\in W^\Lambda$ such that $v R_{\alpha}v'$, $v'R_Xo'$, and $o'R_\square o$. By a similar argument to those in the proof of item 2, we know that $z':=\{\psi ; Y\psi\in o \}$ is consistent. Let $z$ be the $\Lambda$-MCS that includes $z'$. As shown also in item 2, it is the case that $zR_X o$. Let us show that $v R_{\alpha} z$. Let $[\alpha]\phi \in v$. Axiom $(In2)$, axiom $(K)$ for $[\alpha]$, and necessitation for $[\alpha]$ imply that $[\alpha] XY\phi\in v$. Schema $(NA)$ then entails that $[\alpha] X\square Y\phi\in v$, and the assumption that $v R_{\alpha}v'$ then yields that $X\square Y\phi\in v'$. The assumption that $v'R_Xo'$ then yields that $\square Y\phi\in o'$, so that the assumption that $o'R_\square o$ gives that  $Y\phi\in o$, which by construction of $z$ implies that $\phi\in z$. Therefore $v R_{\alpha} z$, which with the previously shown fact that $zR_X o$ implies that $v R_X\circ R_{\alpha} o$. In this way, $R_\square\circ R_X\circ R_{\alpha}\subseteq R_X\circ R_{\alpha}.$

\item $\mathtt{(NAgs)_K}$ It can be shown in the same way as the above item, substituting $[\alpha]$ for $[Ags]$ and using axiom $(NAgs)$ instead of schema $(NA)$. 
\end{itemize}

Finally, we show that $card\left(\mathtt{Choice}_{Ags}\right) \leq n$. Let $w\in W^\Lambda$. Suppose for a contradiction that $card\left(\mathtt{Choice}_{Ags}\right) > n$. Take pairwise different $c_0,\dots,c_n\in \mathtt{Choice}_{Ags}^{\overline{w}}$, and take $w_i\in c_i$ for each $0\leq i \leq n$. According to item 5 of Lemma \ref{carajo} below, for every $1\leq i\leq n$, there must exist $\phi_i\in \mathcal{L}_{\textsf{KX}}$ such that $[Ags]\phi_i\in c_i$ and $\phi_i\not\in c_j$ for every $1\leq j\leq n$ such that $j\neq i$. This means that $\bigwedge\limits_{1\leq i \leq n}\lnot\phi_i\in w_0$ ($\star$) and that $[Ags]\phi_1\in w_1$, $(\lnot \phi_1\land[Ags]\phi_2)\in w_2,\dots, (\lnot \phi_1\land\dots \lnot \phi_{n-1}\land [Ags]\phi_n)\in w_n$. This means that $\bigwedge\limits_{1\leq k\leq n}\Diamond \left(\left(\bigwedge\limits_{1\leq i \leq k-1}\lnot \phi_i\right) \land [Ags] \phi_k\right)\in w_0$, which by $(AgsPC_n)$ implies that $\bigvee\limits_{1\leq k\leq n} \phi_k\in w_0$, but this is a contradiction to ($\star$).\footnote{Observe that we can use the same argument to show theorem $(APC)n$ implies that $card\left(\mathtt{Choice}_\alpha\right) \leq n$ for every $\alpha\in Ags$.}

\item Since the axiom system $\Lambda$ includes the $\mathbf{S5}$ axioms for $K_\alpha$ for each $\alpha\in Ags$, we have that $\approx_\alpha$ is an equivalence relation for each $\alpha\in Ags$. We must now verify that $\mathcal{M}$ validates the constraints $(\mathtt{Unif-H})_K$ and $(\mathtt{NoF})_K$.
\begin{itemize}

\item For $(\mathtt{Unif-H})_K$, fix $w_1,w_2\in W^\Lambda$. We assume that there exist $v\in \overline{w_1}$ and $u\in\overline{w_2}$ such that $v\approx_\alpha u$. Take $v'\in\overline{w_1}$.

We will show that $u''=\{\psi ; K_\alpha\psi\in v'\}\cup \{\square\psi ; \square\psi\in u \}$ is consistent. In order to do so, we will first show that $\{\psi ; K_\alpha\psi\in v'\}$ is consistent. Suppose for a contradiction that it is not consistent. Then there exists a set $\{\psi_1,\dots,\psi_n\}$ of formulas of $\mathcal{L}_{\textsf{KX}}$ such that $\{\psi_1,\dots,\psi_n\}\subseteq \{\psi ; K_\alpha\psi\in v'\}$ and $\vdash \psi_1\wedge\dots\wedge \psi_n\to \bot$ (a). Now, the fact that $\{\psi_1,\dots,\psi_n\}\subseteq \{\psi ; K_\alpha\psi\in v'\}$ means that $K_\alpha \psi_i\in v'$ for every $1\leq i \leq n$; by necessitation of $K_\alpha$ and its distributivity over conjunction, we get that (a) implies that $\vdash K_\alpha\psi_1\wedge\dots\wedge K_\alpha\psi_n\to K_\alpha\bot$, but this is a contradiction, since $
v'$ is a $\Lambda$-MCS which, being, includes $K_\alpha\psi_1\wedge\dots\wedge K_\alpha\psi_n$. 

Next, we show that $u''=\{\psi ; K_\alpha\psi\in v'\}\cup \{\square\psi ; \square\psi\in u \}$ is also consistent. Suppose for a contradiction that it is not consistent. Since $\{\psi ; K_\alpha\psi\in v'\}$ and $\{\square\psi ; \square\psi\in u \}$ are consistent, there must exist sets $\{\phi_1,\dots,\phi_n\}$ and $\{\theta_1,\dots,\theta_m\}$ of formulas of $\mathcal{L}_{\textsf{KX}}$ such that $\{\psi_1,\dots,\psi_n\}\subseteq \{\psi ; K_\alpha\psi\in v'\}$, $\square\theta_i\in w_2$ for every $1\leq i\leq m$, and $\vdash(\phi_1\wedge\dots\wedge\phi_n)\wedge (\square\theta_1\wedge\dots\wedge\square\theta_m)\to\bot$ (a). Let $\theta=\theta_1\land\dots \land\theta_m$  and $\phi=\phi_1\land\dots \land\phi_n$. Since $\square$ distributes over conjunction, we have that $\vdash\square\theta\leftrightarrow\square\theta_1\wedge\dots\wedge\square\theta_m$, where it is important to mention that since $u$ is a $\Lambda$-MCS, then $\square\theta\in u$ and $\square \square\theta\in u$ as well ($\star$). In this way, (a) implies that $\vdash\phi\to\lnot \square \theta$ and thus that $\vdash\Diamond\phi\to\Diamond\lnot\square\theta$ (b). Notice that the facts that $K_\alpha\psi_i\in v'$ for every $1\leq i\leq n$, that $K_\alpha$ distributes over conjunction, and that $v'$ is a $\Lambda$-MCS imply that $K_\alpha\psi\in v'$. The fact that $v'\in \overline {w_1}=\overline{v}$ implies that $\Diamond K_\alpha \psi \in v$, so that $(Unif-H)$ entails that $K_\alpha\Diamond\psi\in v$. Now, this last inclusion implies, with our assumption that $v\approx_\alpha u$, that $\Diamond\psi\in u$, which by (b) in turn yields that $\Diamond\lnot\square\theta\in u$, contradicting $(\star)$. Therefore, $u''$ is consistent. 

Finally, let $u'$ be the $\Lambda$-MCS that includes $u''$. It is clear from its construction that $u'\in\overline{u}=\overline{w_2}$ and that $v'\approx_\alpha u'$, 

With this, we have shown that $\mathcal{M}$ validates the constraint $(\mathtt{Unif-H})_K$.

\item $(\mathtt{NoF})_K$ Let $\alpha \in Ags$ and $w,v\in W^\Lambda$ such that $w \approx_\alpha\circ R_X v$ via $y$. By similar arguments to the ones used in item 2 of this proof, we know that $z':=\{\psi ; Y\psi\in v \}$ is consistent and that if we take $z$ to be the $\Lambda$-MCS that includes $z'$, then $zR_X v$.What remains to be shown is that $w\approx_\alpha z$: assume that $K_\alpha\phi\in w$. Axiom $(In2)$, schema $(K)$ for $K_\alpha$, and necessitation for $K_\alpha$ imply that $K_\alpha XY\phi\in w$. Axiom $(NoF)$ then implies that $XK_\alpha Y\phi\in w$. By construction, this implies that $K_\alpha Y\phi\in y$, which in turn implies that $Y\phi\in v$ and thus --by definition-- that $\phi\in z$. Therefore, we have that $z$ is such that $w\approx_\alpha z$ and $zR_X v$, which means that $w R_X\circ \approx_\alpha v$. Therefore, we have that $\approx_\alpha\circ R_X \subseteq R_X \circ \approx_\alpha$.

\end{itemize}



\end{enumerate}

\end{proof}

As is usual with canonical structures, our objective is to prove the so-called \emph{truth} lemma, which says that for every formula $\phi$ of $\mathcal{L}_{\textsf{KX}}$ and every $w\in W^\Lambda$, we have that $
\mathcal{M} ,w\Vdash \phi \ \textnormal{iff} \ \phi\in w.$ This is done by induction on $\phi$, and the inductive step for each modal operator requires previous results (such as the important \emph{existence} lemmas). These previous results are standard (Lemma \ref{carajo} below). 

\begin{lemma}[Existence] \label{carajo} 
Let $\mathcal{M}$ be the canonical structure for $\Lambda_n$. Let $w\in W^\Lambda$. For a given formula $\phi$ of $\mathcal L_{\textsf{KX}}$, the following hold: 
\begin{enumerate}
\item $X\phi \in w$ iff $\phi\in v$ for every $v\in W^\Lambda$ such that $wR_X v$.
\item $Y\phi \in w$ iff $\phi\in v$ for every $v\in W^\Lambda$ such that $wR_Y v$. 
\item  $\square \phi \in w$ iff $\phi\in v$ for every $v\in \overline{w}$. 
\item $[\alpha ]\phi\in w$ iff $\phi\in v$ for every $v\in \overline{w}$ such that $wR_\alpha v$.
\item $[Ags ]\phi\in w$ iff $\phi\in v$ for every $v\in \overline{w}$ such that $wR_{Ags} v$.
\item $K_\alpha\phi\in w$ iff $\phi\in v$ for every $v\in W^\Lambda$ such that $w\approx_\alpha v$.
\end{enumerate}
\end{lemma} 
\begin{proof}
Let $w\in W^\Lambda$, and take $\phi$ of $\mathcal L_{\textsf{KX}}$. All items are shown in the same way. Let $\triangle\in\{X, Y, \square, [\alpha], [Ags], K_\alpha\}$, and let $R_\triangle$ stand for the relation upon which the semantics of $\triangle\phi$ is defined. We will show that $\triangle\phi \in w$ iff $\phi\in v$ for every $v\in W^\Lambda$ such that $wR_\triangle v$.

$(\Rightarrow)$ We assume that $\triangle\phi\in w$. Let $v\in W^{\Lambda}$ such that $wR_\triangle v$. The definition of $R_\triangle$ straightforwardly gives that $\phi\in v$.  

$(\Leftarrow)$ We work by contraposition. Assume that $\triangle\phi\notin w$. We will show that there is a world $v$ in $W^{\Lambda}$ such that $wR_\triangle v$ and such that $\phi$ does not lie within it. For this, let $v'=\{\psi ;\triangle\psi\in w\}$, which is consistent by a similar argument than the one introduced in the proof of the above item: suppose for a contradiction that $v'$ is not consistent; then there exists a set $\{\psi_1,\dots,\psi_n\}$ of formulas of $\mathcal{L}_{\textsf{KX}}$ such that $\{\psi_1,\dots,\psi_n\}\subseteq v'$ and $\vdash \psi_1\wedge\dots\wedge \psi_n\to \bot$ (a); now, the fact that $\{\psi_1,\dots,\psi_n\}\subseteq v'$ means that $\triangle \psi_i\in w$ for every $1\leq i \leq n$; necessitation of $\triangle$ and its distributivity over conjunction yield that (a) implies that $\vdash\triangle\psi_1\wedge\dots\wedge\triangle\psi_n\to\triangle\bot$, but this is a contradiction, since $w$ is a $\Lambda$-MCS which, being closed under conjunction, includes $\triangle\psi_1\wedge\dots\wedge\triangle\psi_n$. Now, we define $v''=v'\cup \{\lnot\phi\}$, and we show that it is also consistent as follows: suppose for a contradiction that it is not consistent; since $v'$ is consistent, we have that there exists a set $\{\psi_1,\dots,\psi_n\}$ of formulas of $\mathcal{L}_{\textsf{KX}}$ such that $\{\psi_1,\dots,\psi_n\}\subseteq v'$ and $\vdash \psi_1\wedge\dots\wedge \psi_n \wedge \lnot\phi\to \bot$, which then implies that $\vdash \psi_1\wedge\dots\wedge \psi_n \to \phi$ (b); due to necessitation of $\triangle$ and its distributivity over conjunction, we get that (b) implies that $\vdash\triangle\psi_1\wedge\dots\wedge\triangle\psi_n\to\triangle\phi$ (b); but notice that since $w$ is a $\Lambda$-MCS, then $\triangle\psi_1\wedge\dots\wedge\triangle\psi_n\in w$, so that (b) and the fact that $w$ is closed under \emph{Modus Ponens} entail that $\triangle\phi\in w$, contradicting the initial assumption that $\triangle\phi\notin w$. Finally, let $v$ be the $\Lambda$-MCS that includes $v''$. It is clear from its construction that $\phi \notin v$ and that $wR_\triangle v$, by definition of $R_\triangle$.\footnote{Observe that in the cases of $\triangle\in \{[\alpha], [Ags]\}$, axioms $(SET)$ and $(GA)$ render that the found $v$ actually lies within $\overline{w}$ (if $\square\theta \in w$, then $[\alpha]\theta\in w$ and $[Ags]\theta\in w$.}

\end{proof}

\begin{lemma}[Truth Lemma] \label{scof}
Let $\mathcal{M}$ be the canonical structure for $\Lambda_n$. For every formula $\phi$ of $\mathcal{L}_{\textsf{KX}}$ and every $w\in W^\Lambda$, we have that $\mathcal{M},w\Vdash \phi \ \textnormal{iff} \ \phi\in w.$
\end{lemma}

\begin{proof}
We proceed by induction on $\phi$. The cases with the Boolean connectives are standard. For formulas involving the modal operators, both directions follow straightforwardly from Lemma \ref{carajo}. 
\end{proof}

\begin{proposition}[Completeness w.r.t. super-additive $n$-models] \label{complete} The proof system $\Lambda_n$ is complete with respect to the class of Kripke-exstit super-additive $n$-models.
\end{proposition}

\begin{proof}
Let $\phi$ be a $\Lambda_n$-consistent formula of $\mathcal{L}_{\textsf{KX}}$. Let $\Phi$ be the $\Lambda_n$-MCS including $\phi$. We have seen that the canonical structure $\mathcal{M}$ is such that $\mathcal{M} , \Phi \Vdash \phi$.
\end{proof}


\subsection{Irreflexive Super-additive $n$-models}

In the next step of our proof of completeness, for each $\Lambda_n$-consistent formula $\phi$ of $\mathcal{L}_{\textsf{KX}}$, we will build a super-additive  $n$-model where the `next' and `last' relations are irreflexive that satisfies $\phi$. First, we introduce some auxiliary terminology.  For every $w\in W$, we take $h[w]:=\{v\in W; wR_Y^*v \mbox{ or } v=w \mbox{ or } wR_X^*v \}$, where $R_Y^*$ and $R_X^*$ are the transitive closures of $R_Y$ and $R_X$, respectively. It is clear that for every $w\in W$, $R_X^*$ restricted to $h[w]$ is a total order on $h[w]$. For each $w\in W$ and $i\in \mathds{N}\cup\{0\}$, we denote by $w^{-j}$ the unique element in $h[w]$ such that $w=w^{-0}R_Yw^{-1}R_Yw^{-2}R_Y\dots w^{-(j-1)}R_Y w^{-j}$. Similarly, we denote by $w^{+j}$ the unique element in $h[w]$ such that $w=w^{+0}R_Xw^{+1}R_Xw^{+2}R_X\dots w^{+(j-1)}R_X w^{+j}$.   

\begin{definition}\label{unrav}

Let $\mathcal{M}=\langle W, R_\square, R_X, R_Y, \mathtt{Choice}, \{\mathtt{\approx}_\alpha\}_{\alpha\in Ags},\mathcal{V} \rangle$ be a Kripke-exstit super-additive $n$-model for $\mathcal{L}_{\textsf{KX}}$. Consider the following \emph{unraveling} variation  $\mathcal{M}^u=\langle W^u, R_\square^u, R_X^u, R_Y^u, \mathtt{Choice}^u, \{\mathtt{\approx}_\alpha^u\}_{\alpha\in Ags},\mathcal{V}^u \rangle$, where 

\begin{itemize}
    \item Let $T^W$ be the set of all finite sequences $w_0,\dots,w_m$ such that $w_i\in W$ and $m\in \mathds{N}\cup\{0\}$. We define $W^u \subseteq T^W\times \{0,1\}$ such that \begin{enumerate}[a)]
        \item $\langle w_0,\dots,w_m,1\rangle\in W^u$ iff $m\geq 0$ and for all $0\leq i\leq m-1$, $w_i R_X w_{i+1}$.
        \item $\langle w_0,\dots,w_m,0\rangle\in W^u$ iff $m>0$ and for all $0\leq i\leq m-1$, $w_i R_Y w_{i+1}$.
    \end{enumerate}
    \item  We define $R_\square^u$ such that $\langle w_0,\dots,w_m, a\rangle R_\square^u \langle v_0,\dots,v_l, b\rangle$ iff either 
	\begin{itemize} 
	\item $a=b=1$, $l=m$, $w_i R_{Ags}v_i $ for every $0 \leq i \leq m-1$, and $w_m R_\square v_m$, or 

	\item $a=b=0$, $l=m$, and $w_m R_\square v_m$.

	\end{itemize}
    \item  We define $R_X^u$ such that $\langle w_0,\dots,w_m,a\rangle R_X^u \langle v_0,\dots,v_l,b\rangle$ iff either 
	\begin{itemize} 
	\item $a=b=1$, $m\geq 0$, $l=m+1$, $v_i =w_i $ for every $0\leq i \leq m$, and $v_l =w_m^{+1}$, or

	\item $a=b=0$, $m>1$, $l=m-1$, $v_i =w_i $ for every $0\leq i \leq l$, and $w_m = v_l^{+1}$, or
	\item $a=0$, $b=1$, $m=1$, $l=0$, $v_0=w_0=w_1^{+1}$.

	\end{itemize}

    \item We define $R_Y^u$ such that $\langle w_0,\dots,w_m,a\rangle R_Y^u \langle v_0,\dots,v_l,b\rangle$ iff either 
	\begin{itemize} 
	\item $a=b=1$, $m>0$, $l=m-1$, $v_i =w_i $ for every $0\leq i \leq l$, and $v_l =w_m^{-1}$, or

	\item $a=b=0$, $m>0$, $l=m+1$, $v_i =w_i $ for every $0\leq i \leq m$, and $w_m = v_l^{-1}$, or
	\item $a=1$, $b=0$, $m=0$, $l=1$,  $w_0=v_0=v_1^{+1}$.

	\end{itemize}
 
     \item  We define $\mathtt{Choice}^u$ so that 
	\begin{itemize}
        \item For $\alpha\in Ags$, we define $R_\alpha^u$ on $W^u$ the following way. 
        
        $\langle w_0,\dots,w_m,a\rangle R_\alpha^u \langle v_0,\dots,v_l,b\rangle$ iff either 
		\begin{itemize}

		\item $a=b=1$, $l=m$, $w_i R_{Ags}v_i $ for every $0 \leq i \leq m-1$, and $w_m R_\alpha v_m$, or 

		\item $a=b=0$, $l=m$, and $w_m R_\alpha v_m$.

		\end{itemize}         
        
        We set $\mathtt{Choice}_\alpha^u=\left\{R_\alpha^u[\langle w_0,\dots,w_m,a\rangle]|\langle w_0,\dots,w_m,a\rangle\in W^u\right\}$.
        
        \item  $\langle w_0,\dots,w_m,a\rangle R_{Ags}^u \langle v_0,\dots,v_l,b\rangle$ iff either
		\begin{itemize}

		\item $a=b=1$, $l=m$, $w_i R_{Ags}v_i $ for every $0 \leq i \leq m-1$, and $w_m R_{Ags} v_m$, or 

		\item $a=b=0$, $l=m$, and $w_m R_{Ags} v_m$.

		\end{itemize}  
        
        We set $\mathtt{Choice}_{Ags}^u=\left\{R_{Ags}^u[\langle w_0,\dots,w_m,a\rangle]|\langle w_0,\dots,w_m,a\rangle\in W^u\right\}$.
    	\end{itemize}

    \item  For $\alpha\in Ags$, we define $\approx_\alpha^u$ such that $\langle w_0,\dots,w_m,a\rangle \approx_\alpha^u \langle v_0,\dots,v_l,b\rangle$ iff  $w_m \approx_\alpha v_l$.
    
    \item We define $\mathcal{V}^M$ such that $\langle w_0,\dots,w_m,a\rangle\in \mathcal{V}^u(p) $ iff $w_m\in \mathcal{V}(p)$.
\end{itemize}
\end{definition}

\begin{proposition}\label{putamadreojalacarajo}
If $\mathcal{M}$ is a super-additive $n$-model for $\mathcal{L}_{\textsf{KX}}$, then $\mathcal{M}^u$ as defined in Definition \ref{unrav} is a super-additive $n$-model for $\mathcal{L}_{\textsf{KX}}$ where $R_X^u$ and $R_Y^u$ are irreflexive. 
\end{proposition}
\begin{proof}

We want to show that $\langle W^u, R_\square^u, R_X^u, R_Y^u, \mathtt{Choice}^u, \{\mathtt{\approx}^u_\alpha\}_{\alpha\in Ags}\rangle$ is a Kripke-exstit  frame, which amounts to showing that the tuple validates the four items in the definition of Kripke-exstit  frames. 
\begin{enumerate}
    \item It is routine to show that $R_\square^u$ is an equivalence relation. 
    \item From definition we can see that the facts that $R_X$ and $R_Y$ are serial and deterministic imply that $R_X^u$ and $R_Y^u$ are serial and deterministic, respectively. Observe that the variation of the traditional unraveling-argument that we use plays an important role in ensuring that $R_Y^u$ is serial. In order to have predecessors for one-element sequences, we introduced a construction that differentiates ascending from descending sequences. Therefore, for $\langle w, 1\rangle \in W^u$, we have that $\langle w, 1\rangle R_Y^u \langle w, w^{-1},0\rangle$.    
    
    It is routine to show that $(\mathtt{Inverse})$ holds, and the unraveling construction ensures then that $R_X^u$ and $R_Y^u$ are irreflexive.

    \item We have to show that $\mathtt{Choice}^u$ fulfills the requirements of Definition \ref{ledbetter} in the case for super-additive frames:
    \begin{itemize}
        \item It is routine to show that $R_\alpha^u$ and $R_{Ags}^u$ are equivalence relations. 
        \item We now have to show the condition of super-additivity, which amounts to showing that for every $\alpha\in Ags$, $R_{Ags}^u\subseteq R_\alpha^u$. This follows straightforwardly from the definition of these relations and the fact that $R_{Ags}\subseteq R_\alpha$ for every $\alpha\in Ags$.
        \item Also from definition of $R_\alpha^u$ we get that $(\mathtt{SET})_K$ holds. 
        \item For $(\mathtt{IA})_K$,  let $\langle w_0,\dots,w_m,a\rangle \in W^u$ and $s:Ags\to\mathcal{P}\left(\overline{\langle w_0,\dots,w_m,a\rangle }\right)$ be a function that maps each $\alpha\in Ags$ to a member of $\mathtt{Choice}^u_\alpha$ included in $\overline{\langle w_0,\dots,w_m,a\rangle }$. We want to show that $\bigcap_{\alpha \in Ags} s(\alpha) \neq \emptyset$. 
        
        For each $\alpha\in Ags$, take $\langle w_{\alpha 0},\dots,w_{\alpha m},a\rangle\in s(\alpha)$. 
        
        We want to show that there exists $\langle v_0,\dots,v_m,a\rangle\in \overline{\langle w_0,\dots,w_m,a\rangle}$ such that $\langle v_0,\dots,v_m,a\rangle R_\alpha^u \langle w_{\alpha 0},\dots,w_{\alpha m},a\rangle$ for every $\alpha\in Ags$. Observe that by definition we get that $w_{\alpha m}\in \overline{w_m}$ for every $\alpha\in Ags$. Since $\mathcal{M}$ satisfies $(\mathtt{IA})_K$, we have that there exists $v_*\in\overline {w_m}$ such that $v_*R_\alpha w_{\alpha m}$ for every $\alpha\in Ags$. We have two cases, according to the value of $a$: \begin{itemize}
            \item (Case $a=1$) 
            Since $\mathcal{M}$ satisfies $(\mathtt{NAgs})_K$,  $v_*^{-i}R_{Ags}w_{\alpha(m-i)}$ for every $1\leq i \leq m$. Therefore, the finite sequence given by $v_*^{-m},\dots,v_*$ is such that $\langle v_*^{-m},\dots,v_*,1\rangle R_\alpha^u \langle w_{\alpha0},\dots, w_{\alpha m},1 \rangle$ for every $\alpha\in Ags$. 
            \item (Case $a=0$) Similarly, the finite sequence $v_*^{+m},\dots,v_*$ is such that  $\langle v_*^{+m},\dots,v_*,0\rangle R_\alpha^u \langle w_{\alpha0},\dots, w_{\alpha m},0 \rangle$ for every $\alpha\in Ags$. 
        \end{itemize}

        \item For $\mathtt{(NA)_K}$, we want to show that for $\alpha \in Ags$, $R_\square^u\circ R_X^u\circ R_{\alpha}^u\subseteq R_X^u\circ R_{\alpha}^u.$ Therefore, let $\alpha\in Ags$ and $\langle w_0,\dots, w_m,a\rangle, \langle v_0,\dots, v_{l},b\rangle\in W^u$ such that $\langle w_0,\dots, w_m,a\rangle R_\square^u\circ R_X^u\circ R_{\alpha}^u \langle v_0,\dots, v_{l},b\rangle$. We have three cases: \begin{enumerate}[i)]
        \item (Case $a=b=1, m\geq 0$) The assumption implies that there exists $\langle w_0',\dots, w_m',1\rangle\in W^u$ such that $\langle w_0,\dots, w_m,1\rangle R_\alpha^u \langle w_0',\dots, w_m',1\rangle $ and $\langle w_0',\dots, w'_{m+1},1 \rangle R_\square^u \langle v_0,\dots, v_{l},1\rangle$.\footnote{Observe that this last thing implies that actually $\langle v_0,\dots, v_{l},1\rangle=\langle v_0,\dots, v_{m+1},1\rangle$.} The first fact yields that $w_m R_\alpha w_m'$ and the second that $w'_m R_{Ags}v_m$ --which implies by super-additivity of $\mathcal{M}$ that  $w'_m R_{\alpha}v_m$. Therefore, transitivity of $R_\alpha$ yields that $w_m R_\alpha v_m$. Since we have that $w_i R_{Ags} w'_{i}R_{Ags}{v_i}$ for every $0\leq i \leq m-1$, then $\langle w_0,\dots, w_m,1 \rangle R_\alpha^u \langle v_0,\dots, v_m,1 \rangle$, which gives us that $\langle w_0,\dots, w_m, 1\rangle R_X^u\circ R_\alpha^u \langle v_0,\dots, v_{m+1},1 \rangle$.
        \item (Case $a=b=0$, $m>1$) The assumption implies that there exists $\langle w_0',\dots, w_m',0\rangle\in W^u$ such that $\langle w_0,\dots, w_m,0\rangle R_\alpha^u \langle w_0',\dots, w_m',0\rangle$ and $\langle w_0',\dots, w'_{m-1},0 \rangle R_\square^u \langle v_0,\dots, v_{l},0\rangle$. The first fact gives us that $w_m R_\alpha w_m'$, and the second implies by $(\mathtt{NAgs})_K$ that $w'_m R_{Ags}v_m$ --which implies that  $w'_m R_{\alpha}v_m$. Therefore, transitivity of $R_\alpha$ yields that $w_m R_\alpha v_m$ so that $\langle w_0,\dots, w_m, 0\rangle R_X^u\circ R_\alpha^u \langle v_0,\dots, v_{m-1},0  \rangle$.
        \item (Case $a=0$, $b=1$, $m=1$) The assumption implies that there exists $\langle w_0', w_1',0\rangle\in W^u$ with $\langle w_0,w_1,0\rangle R_\alpha^u \langle w_0', w_1',0\rangle$ and $\langle w_0',1 \rangle R_\square^u \langle v_0,1\rangle$. The first fact gives us that $w_1 R_\alpha w_1'$, and the second implies by $(\mathtt{NAgs})_K$ that $w'_1 R_{Ags}v_0^{-1}$. --which implies by super-additivity of $\mathcal{M}$ that  $w'_1 R_{\alpha}v_0^{-1}$. Therefore, transitivity of $R_\alpha$ yields that $w_1 R_\alpha v_0^{-1}$, which gives us that $\langle w_0,w_1, 0\rangle R_\alpha^u \langle v_0,v_0^{-1}, 0\rangle$. In turn, this implies that $\langle w_0,w_1, 0\rangle R_X^u\circ R_\alpha^u \langle v_0,1 \rangle$.
        \end{enumerate}

        \item The fact that $\mathcal{M}^u$ satisfies $\mathtt{(NAgs)_K}$ can be shown in an analogous way to the above item. 
        \item Finally, observe that since $R_\alpha$ and $R_{Ags}$ induce partitions of cardinality at most $n$ on every $\overline{w}$ ($w\in W$), we have that for every $\langle w_0,\dots, w_m,a\rangle$, $R_\alpha^u$ and $R_{Ags}^u$ induce  partitions of cardinality at most $n$ on $\overline{\langle w_0,\dots, w_m,a\rangle}$.
    \end{itemize}

    \item It is clear that the $\approx_\alpha^u$, as defined, are equivalence relations. We must now verify that $\mathcal{M}^u$ validates the constraints $(\mathtt{Unif-H})_K$ and $(\mathtt{NoF})_K$. 

\begin{itemize}
        \item For $(\mathtt{Unif-H})_K$, let $\langle w_0,\dots,w_m,a\rangle$, $\langle w_0',\dots,w_l',a'\rangle\in W^u$. 
        
        Assume that there exist $\langle v_0,\dots,v_m,a\rangle\in \overline{\langle w_0,\dots,w_m,a\rangle}$, $\langle u_0,\dots,u_l,a'\rangle\in \overline{\langle w_0',\dots,w_l',a'\rangle}$ such that $\langle v_0,\dots,v_m,a\rangle\approx_\alpha^u \langle u_0,\dots,u_l,a'\rangle$ --which implies that $v_m\approx_\alpha u_l$ ($\star$). Take $\langle x_0,\dots,x_m,a\rangle\in \overline{\langle w_0,\dots,w_m,a\rangle}$. This implies that $x_m\in\overline{w_m}=\overline{v_m}$, so that $(\star)$ and the fact that $\mathcal{M}$ validates $(\mathtt{Unif-H})_K$ entail that there exists $y\in \overline{u_l}$ such that $x_m\approx_\alpha y$. The finite sequence given by $y^{-l},\dots, y^{-1}, y$ is such that it lies within $T^W$, so that $\langle y^{-l},\dots, y^{-1}, y, a'\rangle$ lies  within $\overline{\langle w_0',\dots,w_l',a'\rangle}$, and such that $\langle x_0,\dots,x_m,a\rangle\approx_\alpha^u \langle y^{-l},\dots, y^{-1}, y, a'\rangle$. 

        \item For $(\mathtt{NoF})_K$, let $\langle w_0,\dots,w_m,a\rangle, \langle v_0,\dots, v_l,b\rangle\in W^u$ such that \\ $\langle w_0,\dots,w_m,a\rangle\approx_\alpha ^u\circ  R_X^u \langle v_0,\dots, v_l,b\rangle$. We have three cases with three sub-cases each: 

(Case $a=1, m\geq 0$) The assumption implies that $\langle w_0,\dots, w_{m+1},1\rangle\approx_\alpha^u \langle v_0,\dots, v_l,b \rangle$. We have the following sub-cases: \begin{itemize}
    \item ($b=1$, $l>0$) Since $\mathcal{M}$ validates $(\mathtt{NoF})_K$, we get that $w_m\approx_\alpha v_{l-1}$, so that $\langle w_0,\dots, w_{m},1\rangle\approx_\alpha^u \langle v_0,\dots, v_{l-1},1 \rangle$.
    \item ($b=1$, $l=0$) Since $\mathcal{M}$ validates $(\mathtt{NoF})_K$, we get that $w_m\approx_\alpha v_{0}^{-1}$, so that $\langle w_0,\dots, w_{m},1\rangle\approx_\alpha^u \langle v_0,v_0^{-1},0 \rangle$.
    \item ($b=0$ $l>0$) Since $\mathcal{M}$ validates $(\mathtt{NoF})_K$, we get that $w_m\approx_\alpha v_{l+1}$, so that $\langle w_0,\dots, w_{m},1\rangle\approx_\alpha^u \langle v_0,\dots,  v_{l+1},0 \rangle$.
\end{itemize}

(Case $a=0, m=1$) The assumption implies that $\langle w_0,1\rangle\approx_\alpha^u \langle v_0,\dots, v_l,b \rangle$. We have the following sub-cases: \begin{itemize}
    \item ($b=1$, $l>0$) Since $\mathcal{M}$ validates $(\mathtt{NoF})_K$, we get that $w_1\approx_\alpha v_{l-1}$, so that $\langle w_0,w_1,0\rangle\approx_\alpha^u \langle v_0,\dots, v_{l-1},1 \rangle$.
    \item ($b=1$, $l=0$) Since $\mathcal{M}$ validates $(\mathtt{NoF})_K$, we get that $w_1\approx_\alpha v_{0}^{-1}$, so that $\langle w_0,w_1,0\rangle\approx_\alpha^u \langle v_0,v_0^{-1},0 \rangle$.
    \item ($b=0$ $l>0$) Since $\mathcal{M}$ validates $(\mathtt{NoF})_K$, we get that $w_1\approx_\alpha v_{l+1}$, so that $\langle w_0,w_1,0\rangle\approx_\alpha^u \langle v_0,\dots,  v_{l+1},0 \rangle$.
\end{itemize}

(Case $a=0, m>1$) The assumption implies that $\langle w_0,\dots,w_{m-1},0\rangle\approx_\alpha^u \langle v_0,\dots, v_l,b \rangle$. We have the following sub-cases: \begin{itemize}
    \item ($b=1$, $l>0$) Since $\mathcal{M}$ validates $(\mathtt{NoF})_K$, we get that $w_m\approx_\alpha v_{l-1}$, so that $\langle w_0,\dots,w_{m},0\rangle\approx_\alpha^u \langle v_0,\dots, v_{l-1},1 \rangle$.
    \item ($b=1$, $l=0$) Since $\mathcal{M}$ validates $(\mathtt{NoF})_K$, we get that $w_m\approx_\alpha v_{0}^{-1}$, so that $\langle w_0,\dots,w_{m},0\rangle\approx_\alpha^u \langle v_0,v_0^{-1},0 \rangle$.
    \item ($b=0$ $l>0$) Since $\mathcal{M}$ validates $(\mathtt{NoF})_K$, we get that $w_m\approx_\alpha v_{l+1}$, so that $\langle w_0,\dots,w_{m},0\rangle\approx_\alpha^u \langle v_0,\dots,  v_{l+1},0 \rangle$.
\end{itemize}
    \end{itemize}
\end{enumerate}

\end{proof}
\begin{proposition}\label{amsterdam}
If $\mathcal{M}$ is a super-additive $n$-model for $\mathcal{L}_{\textsf{KX}}$, then $f:\mathcal{M}^u\to \mathcal{M}$ defined by $f\left(\langle w_0,\dots,w_m,a\rangle\right)=w_m$ is a surjective bounded morphism. 
\end{proposition}
\begin{proof}
\begin{itemize}
\item It is clear that $f$ is surjective. The definition of $\mathcal{V}^u$ in the last item of Definition \ref{unrav} ensures that for every $\langle w_0,\dots,w_m,a\rangle\in W^u$, $\langle w_0,\dots,w_m,a\rangle$ and $f\left(\langle w_0,\dots,w_m,a\rangle\right)$ satisfy the same propositional letters.
    \item Let $R\in\{R_X, R_Y, R_\square, R_\alpha, R_{Ags}, \approx_\alpha\}$, and let $R^u$ stand for the corresponding  relation on $\mathcal{M}^u$. Definition \ref{unrav} ensures that if $\langle w_0,\dots,w_m,a\rangle R^u \langle v_0,\dots,v_l,b\rangle$, then $w_m R v_l$.
    \item Assume that $f\left(\langle w_0,\dots,w_m,a\rangle\right) R \ v$ for $R\in\{R_X, R_Y, R_\square, R_\alpha, R_{Ags}, \approx_\alpha\}$ and $v\in W$. We have the following cases: \begin{itemize}
        \item (Case $a=1$, $m>0$) 
    
    For $R\in \{R_\square, R_\alpha, R_{Ags}, \approx_\alpha\}$, we have that $\langle v^{-m},\dots,v^{-1},v,1 \rangle$ is such that $\langle w_0,\dots,w_m,1\rangle R^u \langle v^{-m},\dots,v^{-1},v,1 \rangle$. 
    
    For $R_X$, we have that $\langle w_0,\dots,w_m,1\rangle R^u_X \langle w_0,\dots,w_m,v, 1\rangle$. 
    
    For $R_Y$, we have that $\langle w_0,\dots,w_m,1\rangle R^u_Y \langle w_0,\dots,v, 1\rangle$. 
    \item (Case $a=1$, $m=0$) 
    
    For $R\in \{R_\square, R_\alpha, R_{Ags}, \approx_\alpha\}$, we have that  $\langle w_0,1\rangle R^u \langle v,1 \rangle$. 
    
    For $R_X$, we have that $\langle w_0,1\rangle R^u_X \langle w_0,v, 1\rangle$. 
    
    For $R_Y$, we have that $\langle w_0,1\rangle R^u_Y \langle w_0,v, 0\rangle$.    
    
    \item  (Case $a=0$, $m=1$) 
    
    For $R\in \{R_\square, R_\alpha, R_{Ags}, \approx_\alpha\}$, we have that $\langle w_0,w_1,0\rangle R^u \langle v^{+1},v,0 \rangle$. 
    
    For $R_X$, we have that $w_0=v$ and therefore $\langle w_0,w_1,0\rangle R^u_X \langle v, 0\rangle$. 
    
    For $R_Y$, we have that  $\langle w_0,w_1,0\rangle R^u_Y \langle w_0,w_1,v, 0\rangle$.    
    
   \item (Case $a=0$, $m>1$)
    
    For $R\in \{R_\square, R_\alpha, R_{Ags}, \approx_\alpha\}$, $\langle w_0,\dots,w_m,0\rangle R^u \langle v^{+m},\dots,v^{+1},v,0 \rangle$. 
    
    For $R_X$, we have that $\langle w_0,\dots,w_m,0\rangle R^u_X \langle w_0,\dots,v, 0\rangle$. 
    
    For $R_Y$, we have that $\langle w_0,\dots,w_m,0\rangle R^u_Y \langle w_0,\dots,w_m,v, 0\rangle$.
    \end{itemize}

\end{itemize}
 Therefore, $f:\mathcal{M}^u\to \mathcal{M}$ is a surjective bounded morphism. 
\end{proof}

\begin{proposition}\label{compunr}
The proof system $\Lambda_n$ is complete with respect to the class of Kripke-exstit super-additive $n$-models where the `next' and `last' relations are irreflexive.
\end{proposition}
\begin{proof}
Let $\phi$ be a $\Lambda$-consistent formula of $\mathcal{L}_{\textsf{KX}}$. By Proposition \ref{complete}, we know that there exists a  Kripke-exstit super-additive  $n$-model $\mathcal{M}$ and a world $w$ in its domain such that $\mathcal{M},w\models \phi$. By Proposition \ref{amsterdam} and the invariance of modal satisfaction under bounded morphisms, we then have that $\mathcal{M}^u$ --as defined in Definition \ref{unrav}-- is a such that  $\mathcal{M}^u,\langle w,1\rangle \models \phi$, where by Proposition \ref{putamadreojalacarajo} we know that $\mathcal{M}^u$ is a Kripke-exstit super-additive $n$-model where the `next' and `last' relations are irreflexive. 
\end{proof}

\subsection{Models}

In the final step of our proof of completeness, for each $\Lambda_n$-consistent formula $\phi$ of $\mathcal{L}_{\textsf{KX}}$, we will build a  model that satisfies it. First, we introduce some auxiliary sets and terminology that will allow us to build the  model.

\begin{definition}\label{manana}
 Let $\mathcal{M}=\langle W, R_\square, R_X, R_Y, \mathtt{Choice}, \{\mathtt{\approx}_\alpha\}_{\alpha\in Ags},\mathcal{V} \rangle$ be a  super-additive $n$-model for $\mathcal{L}_{\textsf{KX}}$ where the `next' and `last' relations are irreflexive.
\begin{itemize}
    
    \item We denote by $W/R_\alpha$ the quotient set of $W$ by $R_\alpha$, meaning the partition set of all $R_\alpha$-equivalence-classes in $W$ (the set of all choice cells in $W$). Let $\mathcal{C}\subseteq \prod\limits_{\alpha\in Ags}W/R_\alpha$ be the set of all choice profiles with non-finite intersection, meaning $\mathcal{C}=\{\overrightarrow{c}\in \prod\limits_{\alpha\in Ags}W/R_\alpha;\bigcap\limits_{\alpha\in Ags} (\overrightarrow{c})_\alpha \neq \emptyset \}$. Since we have that the cardinality of the partition $\mathtt{Choice}_{Ags}$ is bounded by $n$, and that for each $\overrightarrow{c}$ there is at least one member of $\mathtt{Choice}_{Ags}$ included in $\bigcap\limits_{\alpha\in Ags} (\overrightarrow{c})_\alpha$, it is clear that $\mathcal{C}$ is finite and also bounded by $n$.
    \item For every $\overrightarrow{c}\in \mathcal{C}$, let $A_{\overrightarrow{c}}$ stand for the set of all $R_{Ags}$-equivalence-classes included in $\bigcap\limits_{\alpha\in Ags}(\overrightarrow{c})_\alpha$. It is clear that the cardinalities of the $A_{\overrightarrow{c}}$ are uniformly bound by $n$, meaning that for every $\overrightarrow{c}\in \mathcal{C}$, $card(A_{\overrightarrow{c}})\leq n$.
    \item For each $\overrightarrow{c}\in \mathcal{C}$, we consider $\{A_{\overrightarrow{c}}^0,\dots, A_{\overrightarrow{c}}^{n-1}\}$ to be an enumeration of $A_{\overrightarrow{c}}$ with cardinality $n$ (so that repetition is possible).  
    \item For every $w\in W$, we take $h[w]:=\{v\in W; wR_Y^*v \mbox{ or } v=w \mbox{ or } wR_X^*v \}$, where $R_Y^*$ and $R_X^*$ are the transitive closures of $R_Y$ and $R_X$, respectively. It is clear that for every $w\in W$, $R_X^*$ restricted to $h[w]$ is a strict total order on $h[w]$. For each $w\in W$, we take $h^-[w]:=\{v\in h[w];wR_Y^* v\}$ and $h^+[w]:=\{v\in h[w];wR_X^* v\}$. Observe that the fact that $R_X$ is irreflexive implies that for every $w\in W$, $h^-[w], \{w\}$, and $h^+[w]$ are pairwise disjoint sets.  
\end{itemize}
\end{definition}

\begin{lemma}\label{func} Let $\mathcal{M}=\langle W, R_\square, R_X, R_Y, \mathtt{Choice}, \{\mathtt{\approx}_\alpha\}_{\alpha\in Ags},\mathcal{V} \rangle$ be a  super-additive  $n$-model for $\mathcal{L}_{\textsf{KX}}$ where the `next' and `last' relations are irreflexive, and let $w,w'\in W$ such that $wR_\square w'$. For every $v\in h^-[w]$, there exists $v'\in h^-[w']$ such that $vR_{Ags} v'$ and viceversa. 
\end{lemma}
\begin{proof}
We show the result for $w$ by induction on $h^-[w]\cup\{w\}$ --and assume an analogous argument for $w'$. For $w^{-1}$, condition $(\mathtt{NAgs})_K$ ensures that $w^{-1}R_{Ags} w'^{-1}$. Suppose that the property holds for $w^{-j}$ so that there exists $v'\in h^-[w']$ with $w^{-j}R_{Ags} v'$. Since $R_{Ags}\subseteq R_\square$, this means that $w^{-j}R_{\square} v'$ as well. Since $R_{Ags}$ is reflexive, this implies that $w^{-(j+1)}R_\square\circ R_X\circ R_{Ags} v'$, which by $(\mathtt{NAgs})_K$ implies that $w^{-(j+1)}R_{Ags}v'^{-1}$. 
\end{proof}

\begin{definition}\label{matriz} Let $\mathcal{M}=\langle W, R_\square, R_X, R_Y, \mathtt{Choice}, \{\mathtt{\approx}_\alpha\}_{\alpha\in Ags},\mathcal{V} \rangle$ be a  super-additive $n$-model for $\mathcal{L}_{\textsf{KX}}$ where the `next' and `last' relations are irreflexive. We define a structure $\mathcal{M}^M:=\langle W^M, R_\square^M, R_X^M, R_Y^M, \mathtt{Choice}^M, \{\mathtt{\approx}^M_\alpha\}_{\alpha\in Ags},\mathcal{V}^M \rangle$ and show that it is a  model. We use an argument found that mirrors the technique of \cite{schwarzentruber2012complexity} and \cite{lorini2013temporal}. We define $\mathcal{M}^M$ the following way:
\begin{itemize}
    \item The domain $W^M$ is given by \scriptsize  \[W^M:=\left\{\left(\overrightarrow{c}, \vec{f}, w\right); \begin{array}{l} \overrightarrow{c}\in \mathcal{C}, w\in \bigcap\limits_{\alpha\in Ags}\left(\overrightarrow{c}\right)_\alpha \\ \vec{f}:h[w] \to\prod\limits_{\alpha\in Ags}\{0,\dots, n-1\},\\ \mbox{for every } v\in h[w], v\in A^{\left(\sum_{\alpha\in Ags}\left(\vec{f}(v)\right)_\alpha \right)\mod n}_{\overrightarrow{c}_{v}}\\ \vec{f}(v)= \vec{f}(v') \mbox{ if } vR_{Ags}v'\end{array} \right\}.\footnote{Observe that $\overrightarrow{c}_{v}$ is taken to be the unique choice profile in $\mathcal{C}$ such that $v\in \bigcap\limits_{\alpha\in Ags}\left(\overrightarrow{c}_{v}\right)_\alpha$.}\]
    \normalsize
    \item We define $R_\square^M$ such that $\left(\overrightarrow{c}, \vec{f}, w\right)R_\square^M\left(\overrightarrow{c}', \vec{f'}, w'\right)$ iff $wR_\square w'$ and for all $v\in h^-[w], v'\in h^-[w']$ such that $v R_{Ags} v'$, we have that $\vec{f}(v)=\vec{f'}(v')$.
    \item  We define $R_X^M$ such that $\left(\overrightarrow{c}, \vec{f}, w\right)R_X^M\left(\overrightarrow{c}', \vec{f'}, w'\right)$ iff $wR_X w'$ and  $ \vec{f'}= \vec{f}$.
    \item  We define $R_Y^M$ such that $\left(\overrightarrow{c}, \vec{f}, w\right)R_Y^M\left(\overrightarrow{c}', \vec{f'}, w'\right)$ iff $wR_Y w'$ and $\vec{f'}= \vec{f}$.
    \item  We define $\mathtt{Choice}^M$ so that \begin{itemize}
        \item For $\alpha\in Ags$, we define $R_\alpha^M$ on $W^M$ the following way. 
        
        Let $\left(\overrightarrow{c}, \vec{f}, w\right),\left(\overrightarrow{c}', \vec{f'}, w'\right)\in W^M$. We have that $\left(\overrightarrow{c}, \vec{f}, w\right)R_\alpha^M\left(\overrightarrow{c}', \vec{f'}, w'\right)$ iff \begin{itemize}
            \item $\left(\overrightarrow{c}'\right)_\alpha=\left(\overrightarrow{c}\right)_\alpha$ (which means that $wR_\alpha w'$),
            \item For all $v\in h^-[w], v'\in h^-[w']$ such that $v R_{Ags} v'$, we have that $\vec{f}(v)=\vec{f'}(v')$,
            \item  $\left(\vec{f}'(w')\right)_\alpha=\left(\vec{f}(w)\right)_\alpha$. 
        \end{itemize}

        In this way, we set $\mathtt{Choice}_\alpha^M=\left\{R_\alpha^M\left[\left(\overrightarrow{c}, \vec{f}, w\right)\right]|\left(\overrightarrow{c}, \vec{f}, w\right)\in W^M\right\}$.
        
        \item $\left(\overrightarrow{c}, \vec{f}, w\right)R_{Ags}^M\left(\overrightarrow{c}', \vec{f'}, w'\right)$ iff \begin{itemize}
            \item $ \overrightarrow{c}'=\overrightarrow{c}$ (which implies that $wR_{Ags} w'$),
            \item  For all $v\in h^-[w], v'\in h^-[w']$ such that $v R_{Ags} v'$, we have that $\vec{f}(v)=\vec{f'}(v')$,
            \item $\vec{f}(w)=\vec{f'}(w')$.
        \end{itemize}   
        
        In this way, we set $\mathtt{Choice}_{Ags}^M=\left\{R_{Ags}^M\left[\left(\overrightarrow{c}, \vec{f}, w\right)\right]|\left(\overrightarrow{c}, \vec{f}, w\right)\in W^M\right\}$.
    \end{itemize}
    \item For $\alpha\in Ags$, we define $\approx_\alpha^M$ such that $\left(\overrightarrow{c}, \vec{f}, w\right)\approx_\alpha^M\left(\overrightarrow{c}', \vec{f'}, w'\right)$ iff $w\approx_\alpha w'$.
    \item We define $\mathcal{V}^M$ such that $\left(\overrightarrow{c}, \vec{f}, w\right)\in \mathcal{V}^M(p) $ iff $w\in \mathcal{V}(p)$.
\end{itemize}
\end{definition}

\begin{proposition}\label{putamadreojala}
If $\mathcal{M}$ is a super-additive $n$-model for $\mathcal{L}_{\textsf{KX}}$ where the `next' and `last' relations are irreflexive, then $\mathcal{M}^M$ as defined in Definition \ref{matriz} is an $n$-model for $\mathcal{L}_{\textsf{KX}}$.
\end{proposition}
\begin{proof}

We want to show that $\langle W^M, R_\square^M, R_X^M, R_Y^M, \mathtt{Choice}^M, \{\mathtt{\approx}^M_\alpha\}_{\alpha\in Ags}\rangle$ is a Kripke-exstit  frame, which amounts to showing that the tuple validates the four items in the definition of Kripke-exstit  frames. 

\begin{enumerate}

\item It is routine to show that $R_\square^M$ is an equivalence relation. Observe that symmetry and transitivity follow from symmetry and transitivity of $R_{Ags}$, respectively.

\item Let us show that $R_X^M$ and $R_Y^M$ are serial, deterministic, and irreflexive. We show it for $R_X^M$ and assume the same arguments for $R_Y^M$. Let $\left(\overrightarrow{c}, \vec{f}, w\right)\in W^M$. Since $R_X$ is serial and deterministic, we know that there exists a unique $w^{+1}\in W$ such that $wR_Xw^{+1}$. Since $\mathtt{Choice}_{\alpha}$ is a partition of $W$ for every $\alpha$, there is a unique $\overrightarrow{c}_{w^{+1}}\in \mathcal{C}$ such that $w^{+1}\in \bigcap\limits_{\alpha\in Ags}\left(\overrightarrow{c}_{w^{+1}}\right)_\alpha$. Observe that $w^{+1}\in h[w]$ and that $w^{+1}\in A^{\left(\sum_{\alpha\in Ags}\left(\vec{f}\left(w^{+1}\right)\right)_\alpha \right)\mod n}_{\overrightarrow{c}_{w^{+1}}}$. This means that the tuple $\left(\overrightarrow{c}_{w^{+1}}, \vec{f}, w^{+1}\right)$ is a member of  $W^M$, and it is the only member of $W^M$ such that  $\left(\overrightarrow{c},\vec{f}, w\right)R_X^M\left(\overrightarrow{c}_{w^{+1}}, \vec{f}, w^{+1}\right)$. Moreover, since we know by assumption that $w^{+1}\neq w$, we have that each successor is different from its predecessor, which renders that $R_X^M$ is irreflexive. In an analogous way, we can see that $R_Y^M$ is also irreflexive. 

\begin{itemize}
   \item $(\mathtt{Inverse})$ Let $\left(\overrightarrow{c}, \vec{f}, w\right)\in W^M$, and let $\left(\overrightarrow{c}_{w^{-1}}, \vec{f}, w^{-1}\right)$, $\left(\overrightarrow{c}_{w^{+1}}, \vec{f}, w^{+1}\right)$ be the unique members of $W^M$ such that $\left(\overrightarrow{c}, \vec{f}, w\right)R_Y^M\left(\overrightarrow{c}_{w^{-1}}, \vec{f}, w^{-1}\right)$ and $\left(\overrightarrow{c}, \vec{f}, w\right)R_X^M\left(\overrightarrow{c}_{w^{+1}}, \vec{f}, w^{+1}\right)$. 
   
   Assume that $\left(\overrightarrow{c}_{w^{-1}}, \vec{f}, w^{-1}\right)R_X^M\left(\overrightarrow{c}', \vec{f'}, w'\right)$. Since $\mathcal{M}$ validates $(\mathtt{Inverse})$, it is clear that $w'=w$. This last thing implies that $\overrightarrow{c}'=\overrightarrow{c}$ --by definition of $W^M$. Now, by definition of $R_X^M$, the assumption that \\ $\left(\overrightarrow{c}_{w^{-1}}, \vec{f}, w^{-1}\right)R_X^M\left(\overrightarrow{c}', \vec{f'}, w'\right)$ implies that $\vec{f'}=\vec{f}$. Therefore, we have that $\left(\overrightarrow{c}', \vec{f'}, w'\right)= \left(\overrightarrow{c}, \vec{f}, w\right)$, so  $R_X\circ R_Y=Id$. 
   
   A similar argument yields that if $\left(\overrightarrow{c}_{w^{+1}}, \vec{f}, w^{+1}\right)R_Y^M \left(\overrightarrow{c}', \vec{f'}, w'\right)$, then \\ $\left(\overrightarrow{c}', \vec{f'}, w'\right)= \left(\overrightarrow{c}, \vec{f}, w\right)$ so that $R_Y\circ R_X=Id$.
    
\end{itemize}

\item We have to show that $\mathtt{Choice}^M$ fulfills the requirements of Definition \ref{ledbetter} and substantiates the fact that the frame underlying $\mathcal{M}$ is an actual frame. \begin{itemize}
    \item It is routine to show that for every $\alpha\in Ags$, $R_\alpha^M$ is an equivalence relation. Therefore, $\mathtt{Choice}_\alpha^M$ is indeed a partition of $W^M$ for every $\alpha \in Ags$.
    \item Similarly, $R_{Ags}^M$ is an equivalence relation. Therefore, $\mathtt{Choice}_{Ags}$ is indeed a partition of $W^M$.
    \item We have to show that for every $\left(\overrightarrow{c}, \vec{f}, w\right)\in W^M$,\[\mathtt{Choice}^M_{Ags}\left(\left(\overrightarrow{c}, \vec{f}, w\right)\right)=\bigcap\limits_{\alpha\in Ags}\mathtt{Choice}^M_\alpha\left(\left(\overrightarrow{c}, \vec{f}, w\right)\right).\]
    This amounts to showing that $R_{Ags}^M=\bigcap\limits_{\alpha\in Ags}R_\alpha^M$. 
    
    For the $(\subseteq)$ inclusion, assume that $\left(\overrightarrow{c}, \vec{f}, w\right)R_{Ags}^M\left(\overrightarrow{c}', \vec{f'}, w'\right)$. This implies that $wR_{Ags }w'$, so the fact that $\mathcal{M}$ is super-additive yields that $wR_{\alpha}w'$ for every $\alpha\in Ags$. Our assumption also implies by definition that $\overrightarrow{c}'=\overrightarrow{c}$, that for all $v\in h^-[w], v'\in h^-[w']$ such that $v R_{Ags} v'$ we have that $\vec{f}(v)=\vec{f'}(v')$, and that $\vec{f}(w)=\vec{f}(w')$. This straightforwardly implies that $ \left(\overrightarrow{c}'\right)_\alpha=\left(\overrightarrow{c}\right)_\alpha$ and that $\left(\vec{f}'(w')\right)_\alpha=\left(\vec{f}(w)\right)_\alpha$ \emph{for every} $\alpha\in Ags$. Therefore, we have that $\left(\overrightarrow{c}, \vec{f}, w\right)R_{\alpha}^M\left(\overrightarrow{c}', \vec{f'}, w'\right)$ for every $\alpha\in Ags$, so that $R_{Ags}^M\subseteq\bigcap\limits_{\alpha\in Ags}R_\alpha^M$.
    
    For the $(\supseteq)$ inclusion, assume that $\left(\overrightarrow{c}, \vec{f}, w\right)R_{\alpha}^M\left(\overrightarrow{c}', \vec{f'}, w'\right)$ for every $\alpha\in Ags$. This implies that $wR_{\alpha}w'$ for every $\alpha \in Ags$, that $\left(\overrightarrow{c}'\right)_\alpha=\left(\overrightarrow{c}\right)_\alpha$ for every $\alpha \in Ags$, that for all $v\in h^-[w], v'\in h^-[w']$ such that $v R_{Ags} v'$ we have that $\vec{f}(v)=\vec{f'}(v')$, and that $\left(\vec{f}'(w')\right)_\alpha=\left(\vec{f}(w)\right)_\alpha$ for every $\alpha\in Ags$. Thus, we have that $\overrightarrow{c}'=\overrightarrow{c}$ and that $\vec{f}'(w')=\vec{f}(w)$, so we get that $\left(\overrightarrow{c}, \vec{f}, w\right)R_{Ags}^M\left(\overrightarrow{c}', \vec{f'}, w'\right)$. With this we have shown that $R_{Ags}^M\supseteq\bigcap\limits_{\alpha\in Ags}R_\alpha^M$.

    \item $(\mathtt{SET})_K$ Straightforward from the definitions of the $R_\alpha^M$, we have that $R_\alpha^M\subseteq R_\square^M$ for every $\alpha\in Ags$.
\item  $(\mathtt{IA})_K$ Let $\left(\overrightarrow{c}, \vec{f}, w\right)\in W^M$, and let $s:Ags\to\mathcal{P}\left(\overline{\left(\overrightarrow{c}, \vec{f}, w\right)}\right)$ be a function that maps each $\alpha\in Ags$ to a member of $\mathtt{Choice}^M_\alpha$ included in $\overline{\left(\overrightarrow{c}, \vec{f}, w\right)}$. We want to show that $\bigcap_{\alpha \in Ags} s(\alpha) \neq \emptyset$. 

 For each $s(\alpha)$, take $\left(\overrightarrow{c}^\alpha, \vec{f}^\alpha, w_\alpha\right)\in s(\alpha)$. Observe then that for a fixed $\alpha\in Ags$, the fact that $\left(\overrightarrow{c}^\alpha, \vec{f}^\alpha, w_\alpha\right)\in \overline{\left(\overrightarrow{c}, \vec{f}, w\right)}$ implies that $w_\alpha\in \overline{w}$ and that for every $x\in h^-[w_\alpha], x' \in h^-[w]$ such that $xR_{Ags}x'$, $\vec{f}^\alpha(x)=\vec{f}({x'})$.

Now, since $\mathcal{M}$ validates $(\mathtt{IA})_K$, we have that $\bigcap\limits_{\alpha\in Ags}\mathtt{Choice}_\alpha(w_\alpha)\neq \emptyset$. Let $m=card(Ags)$. We define \\ $\overrightarrow{d}:=\left(\mathtt{Choice}_{\alpha_1}(w_{\alpha_1}),\dots,\mathtt{Choice}_{\alpha_{m}}(w_{\alpha_m})\right)$. It is clear that $\overrightarrow{d}\in \mathcal{C}$ and that $A_{\overrightarrow{d}}$ is non-empty. 

 Consider the number given by $N:=\left(\sum\limits_{\alpha\in Ags}\left( \vec{f}^\alpha\left(w_\alpha\right)\right)_\alpha\right) \mod n$. 
 We know that $A_{\overrightarrow{d}}^N$ is non-empty, so take $v_*\in A_{\overrightarrow{d}}^N$ --where we recall that $A_{\overrightarrow{d}}^N$ is an $R_{Ags}$-equivalence class included in $\bigcap\limits_{\alpha\in Ags}\mathtt{Choice}_\alpha(w_\alpha)$. It is clear that $v_*R_\square w$ and that for every $\alpha\in Ags$, $v_*R_\alpha w_\alpha$.
 
 We want to `build' a function $\vec{f}^*$ such that $\left(\overrightarrow{d}, \vec{f}^*,v_*\right)\in \overline{\left(\overrightarrow{c}, \vec{f}, w\right)}$ and such that $\left(\overrightarrow{d}, \vec{f}^*,v_*\right)\in s(\alpha)$ for every $\alpha\in Ags$.

In order to do that, we first observe that $\vec{f}^*$ should assign a vector to each $u\in h[v_*]$. We will define $\vec{f}^*$ by parts and then show that our definition yields that $\left(\overrightarrow{d}, \vec{f}^*,v_*\right)$ satisfies the conditions that we want.
For every $u\in h^-[v_*]$, we know by Lemma \ref{func} that there exists $u'\in h^-[w]$ such that $uR_{Ags}u'$. Therefore, for every $u\in h^-[v_*]$, we take $\vec{f}^*(u)= \vec{f}(u')$. Observe that this part of $f_*$ is well-defined on account of the fact that for every $x,y\in h^-[w]$ such that $xR_{Ags}y$ we have that $\vec{f}(x)=\vec{f}(y)$ by definition of $\vec{f}$. For $v_*$, we set  $\vec{f}^*(v_*)=\left(\left(\vec{f}^{\alpha_1}\left(w_{\alpha_1}\right)\right)_{\alpha_1},\dots,\left(\vec{f}^{\alpha_m}\left(w_{\alpha_m}\right)\right)_{\alpha_m}\right)$. Finally, for each $u\in h^+[v_*]$, we take $\overrightarrow{d}_{u}$ to be the unique choice profile in $\mathcal{C}$ such that $u\in \bigcap\limits_{\alpha\in Ags}\left(\overrightarrow{d}_{u}\right)_\alpha$, and we consider the elements in $A_{\overrightarrow{d}_{u}}$. It is clear that for each $u\in h^+[v_*]$, there is a unique element in $A_{\overrightarrow{d}_{u}}$ that includes $u$. Let $N_{u}$ be the index of said element in the enumeration $\left\{A_{\overrightarrow{d}_{u}}^0,\dots, A_{\overrightarrow{d}_{u}}^{n-1}\right\}$. If we take $(N_{u},0,\dots,0)\in \prod\limits_{\alpha\in Ags}\{0,\dots,n-1\}$, then it is clear that $u\in A_{\overrightarrow{d}_{u}}^{\left(\sum_{\alpha\in Ags}(N_{u},0,\dots,0)_\alpha\right) \mod n}$, where $(N_{u},0,\dots,0)_\alpha$ stands for the $\alpha$-entry in the vector $(N_{u},0,\dots,0)$. Therefore, for each $u\in h^+[v_*]$, we define $\vec{f}^*(u)=(N_{u},0,\dots,0)$. In order for $f^*$ to be well-defined, \emph{it is crucial} that $R_X$ is irreflexive, for this condition implies that $h^-[v_*], \{v_*\}$, and $h^+[v_*]$ are pairwise disjoint sets, as we had mentioned in Definition \ref{manana}. 

Observe that the definition of $\vec{f}^*$ implies that $\left(\overrightarrow{d}, \vec{f}^*,v_*\right)\in W^M$. Let us show that it also implies that (a) $\left(\overrightarrow{d}, \vec{f}^*,v_*\right)\in \overline{\left(\overrightarrow{c}, \vec{f}, w\right)}$, and that (b) $\left(\overrightarrow{d}, \vec{f}^*,v_*\right)\in s(\alpha)$ for every $\alpha\in Ags$. 

For (a), observe that we took $v_*\in \overline{w}$ and that we set $\vec{f}^*(u)= \vec{f}(u')$ for every $u\in h^-[v_*], u'\in h^-[w]$ such that $uR_{Ags} u'$, so that $\left(\overrightarrow{d}, \vec{f}^*,v_*\right)\in \overline{\left(\overrightarrow{c}, \vec{f}, w\right)}$.

For (b), we have to prove the three items in the definition of $R_\alpha^M$. Fix $\alpha\in Ags$. First, observe that $v_*R_\alpha w_\alpha$, which implies that $\left(\overrightarrow{c}^\alpha\right)_\alpha=\left(\overrightarrow{d}\right)_\alpha$. 

We also know that for every $x_\alpha\in h^-[w_\alpha]$ and $x \in h^-[w]$ such that $x_\alpha R_{Ags}x$, $\vec{f}^\alpha(x_\alpha)=\vec{f}({x})$. Let $u\in h^-[v_*], u_\alpha\in h^-[w_\alpha]$ such that $uR_{Ags}u_\alpha$. We know that $\vec{f}^*(u)=\vec{f}({u'})$ for some (and for every) $u'\in h^-[w]$ such that $uR_{Ags}u'$. The facts that $uR_{Ags}u'$ and $uR_{Ags}u_\alpha$ imply by euclideanity of $R_{Ags}$ that $u'R_{Ags} u_\alpha$. Therefore, we have that the fact that $\left(\overrightarrow{c}^\alpha, \vec{f}^\alpha, w_\alpha\right)\in \overline{\left(\overrightarrow{c}, \vec{f}, w\right)}$ implies that $\vec{f}^\alpha(u_\alpha)=\vec{f}({u'})$. Thus, we have that $\vec{f}^*(u)=\vec{f}({u'})=\vec{f}^\alpha({u_\alpha})$. 

Finally, we have that $\left(\vec{f}^*(v_*)\right)_\alpha=\left(\vec{f}^\alpha(w_\alpha)\right)_\alpha$ by construction. 

Therefore, we have shown that for every $\alpha \in Ags$, \\ $\left(\overrightarrow{d}, \vec{f}^*,v_*\right)R_\alpha^M \left(\overrightarrow{c}^\alpha, \vec{f}^\alpha,w_\alpha\right)$, which means that $\left(\overrightarrow{d}, \vec{f}^*,v_*\right)\in s(\alpha)$ for every $\alpha\in Ags$.

\end{itemize}

 \begin{itemize}

\item $\mathtt{(NA)_K}$ We want to show that for $\alpha \in Ags$, $R_\square^M\circ R_X^M\circ R_{\alpha}^M\subseteq R_X^M\circ R_{\alpha}^M.$ Let $\alpha \in Ags$ and $\left(\overrightarrow{c}, \vec{f}, w\right),\left(\overrightarrow{d}, \vec{g}, v\right)\in W^M$ such that $\left(\overrightarrow{c}, \vec{f}, w\right) R_\square^M\circ R_X^M\circ R_{\alpha}^M \left(\overrightarrow{d}, \vec{g}, v\right) $, which means that there exists $\left(\overrightarrow{c}', \vec{f}', w'\right)$ such that $\left(\overrightarrow{c}, \vec{f}, w\right)R_\alpha^M\left(\overrightarrow{c}', \vec{f}', w'\right)$ and such that $\left(\overrightarrow{c}'_{w'^{+1}}, \vec{f}', w'^{+1}\right)R_\square^M \left(\overrightarrow{d}, \vec{g}, v\right)$. We want to show that $\left(\overrightarrow{c}, \vec{f}, w\right) R_\alpha^M \left(\overrightarrow{d}_{v^{-1}}, \vec{g}, v^{-1}\right)$.

First, observe that the fact that $\left(\overrightarrow{c}'_{w'^{+1}}, \vec{f}', w'^{+1}\right)R_\square^M \left(\overrightarrow{d}, \vec{g}, v\right)$ implies that $w'^{+1}R_\square v$, which by $(\mathtt{NAgs})_K$ implies that $w'R_{Ags} v^{-1}$. It also implies that for every $x\in h^-[w'^{+1}]$ and $y\in h^-[v]$ such that $xR_{Ags}y$, $\vec{f'}(x)=\vec{g}(y)$. In particular, this means that $\overrightarrow{c}'=\overrightarrow{d}_{v^{-1}}$ and that $\vec{f'}(w')=\vec{g}(v^{-1})$. Observe then that these last two facts, coupled with our assumption that $\left(\overrightarrow{c}, \vec{f}, w\right)R_\alpha^M\left(\overrightarrow{c}', \vec{f}', w'\right)$, yield that two of the items in the definition of $R_\alpha^M$ (the first and third, respectively) hold between $\left(\overrightarrow{c}, \vec{f}, w\right) $ and $\left(\overrightarrow{d}_{v^{-1}}, \vec{g}, v^{-1}\right)$, namely that $\left(\overrightarrow{c}\right)_\alpha=\left(\overrightarrow{c}'\right)_\alpha=\left(\overrightarrow{d}_{v^{-1}}\right)_\alpha$, and that $\left(\vec{f}(w)\right)_\alpha=\left(\vec{f'}(w')\right)_\alpha=\left(\vec{g}(v^{-1})\right)_\alpha$. 

For the second item we have to show that for every $x\in h^-[w]$ and $y\in h^-[v^{-1}]$ such that $xR_{Ags}y$, $\vec{f}(x)=\vec{g}(y)$. So let $x\in h^-[w]$ and $y\in h^-[v^{-1}]$ such that $xR_{Ags}y$.

From the assumption that $\left(\overrightarrow{c}, \vec{f}, w\right)R_\alpha^M\left(\overrightarrow{c}', \vec{f}', w'\right)$ we get that there exists $x'\in h^-[w']$ such that $xR_{Ags}x'$ and  $\vec{f}(x)=\vec{f}'(x')$, where it is clear that $x'\in h^-[w']$ implies that $x'\in h^-[w'^{+1}]$ as well. 
Now, observe that the fact that $y\in h^-[v^{-1}]$ implies that $y\in h^-[v]$. In turn, the facts that $xR_{Ags}y$ and that $xR_{Ags}x'$ imply that $yR_{Ags}x'$. Therefore, we get that $x'\in h^-[w'^{+1}]$ and $y\in h^-[v]$ are such that $yR_{Ags}x'$, so that the fact that $\left(\overrightarrow{c}'_{w'^{+1}}, \vec{f}', w'^{+1}\right)R_\square^M \left(\overrightarrow{d}, \vec{g}, v\right)$ implies that  $\vec{f}'(x')=\vec{g}(y)$. In turn, this yields that $\vec{f}(x)=\vec{f}'(x')=\vec{g}(y)$, giving us what we wanted.     

\item $\mathtt{(NAgs)_K}$ It can be shown in the same way as the above item.
\item Finally, observe that since $R_\alpha$ and $R_{Ags}$ both induce partitions of cardinality at most $n$ on every $\overline{w}$ ($w\in W$), we have that for every $\left(\overrightarrow{c}, \vec{f}, w\right)$, $R_\alpha^M$ and $R_{Ags}^M$ induce  partitions of cardinality at most $n$ on $\overline{\left(\overrightarrow{c}, \vec{f}, w\right)}$.
\end{itemize}

\item Since $\approx_\alpha$ is an equivalence relation on $W$ for every $\alpha\in Ags$, we have that $\approx_\alpha^M$ is also an equivalence relation on $W^M$.  We must now verify that $\mathcal{M}^M$ validates the constraints $(\mathtt{Unif-H})_K$ and $(\mathtt{NoF})_K$.
\begin{enumerate}[i)]

\item For $(\mathtt{Unif-H})_K$, fix $\left(\overrightarrow{c}_1, \vec{f}_1, w_1\right),\left(\overrightarrow{c}_2, \vec{f}_2, w_2\right)\in W^M$. Assume that there exist $\left(\overrightarrow{c}, \vec{f}, w\right)\in \overline{\left(\overrightarrow{c}_1, \vec{f}_1, w_1\right)}$ and $\left(\overrightarrow{d}, \vec{g}, v\right)\in\overline{\left(\overrightarrow{c}_2, \vec{f}_2, w_2\right)}$ such that $\left(\overrightarrow{c}, \vec{f}, w\right)\approx_\alpha^M \left(\overrightarrow{d}, \vec{g}, v\right)$. Take $\left(\overrightarrow{c}', \vec{f}', w'\right)\in\overline{\left(\overrightarrow{c}_1, \vec{f}_1, w_1\right)}$. We want to show that there exists $\left(\overrightarrow{d}', \vec{g}', v'\right)\in \overline{\left(\overrightarrow{c}_2, \vec{f}_2, w_2\right)}$ such that $\left(\overrightarrow{c}', \vec{f}', w'\right)\approx_\alpha^M \left(\overrightarrow{d}', \vec{g}', v'\right)$.

The facts that $\left(\overrightarrow{c}, \vec{f}, w\right)\in \overline{\left(\overrightarrow{c}_1, \vec{f}_1, w_1\right)}$, that $\left(\overrightarrow{d}, \vec{g}, v\right)\in\overline{\left(\overrightarrow{c}_2, \vec{f}_2, w_2\right)}$, and that $\left(\overrightarrow{c}', \vec{f}', w'\right)\in\overline{\left(\overrightarrow{c}_1, \vec{f}_1, w_1\right)}$  imply by definition that $wR_\square w_1$, that $v R_\square w_2$, and that $w'R_\square w_1$. Also by definition we have that our assumption implies that $w\approx_\alpha v$. Since $\mathcal{M}$ validates $(\mathtt{Unif-H})_K$, we then get that there exists $v'\in \overline{w_2} $ such that $w'\approx_\alpha v'$. In a similar way to what we did to show that $(\mathtt{IA})_K$ holds, we will build a function $\vec{g}'$ such that $\left(\overrightarrow{d}', \vec{g}',v'\right)\in \overline{\left(\overrightarrow{c}_2, \vec{f}_2, w_2\right)}$ and such that $\left(\overrightarrow{c}', \vec{f}', w'\right)\approx_\alpha^M \left(\overrightarrow{d}', \vec{g}', v'\right)$. 

Firstly, let $\overrightarrow{d}'$ be the unique choice profile such that in $\mathcal{C}$ such that $v'\in \bigcap\limits_{\alpha\in Ags}\left(\overrightarrow{d}\right)_\alpha$

 Again we define $\vec{g}'$ by parts so that it satisfies the conditions that we want.
For every $u\in h^-[v']$, we know by Lemma \ref{func} that there exists $u'\in h^-[w_2]$ such that $uR_{Ags}u'$. Therefore, for every $u\in h^-[v']$, we take $\vec{g}'(u)= \vec{f}_2(u')$. Observe that this part of $\vec{g}'$ is well-defined on account of the fact that for every $x,y\in h^-[w_2]$ such that $xR_{Ags}y$ we have that $\vec{f}_2(x)=\vec{f}_2(y)$. For each $u\in \{v'\}\cup h^+[v']$, we take $\overrightarrow{d}_{u}$ to be the unique choice profile in $\mathcal{C}$ such that $u\in \bigcap\limits_{\alpha\in Ags}\left(\overrightarrow{d}_{u}\right)_\alpha$,\footnote{If $u=v'$, then $\overrightarrow{d}_{u}=\overrightarrow{d}'$.} and we consider the elements in $A_{\overrightarrow{d}_{u}}$. It is clear that for each $u\in \{v'\}\cup h^+[v']$, there is a unique element in $A_{\overrightarrow{d}_{u}}$ that includes $u$. Let $N_{u}$ be the index of said element in the enumeration $\left\{A_{\overrightarrow{d}_{u}}^0,\dots, A_{\overrightarrow{d}_{u}}^{n-1}\right\}$. If we take $(N_{u},0,\dots,0)\in \prod\limits_{\alpha\in Ags}\{0,\dots,n-1\}$, then it is clear that $u\in A_{\overrightarrow{d}_{u}}^{\left(\sum_{\alpha\in Ags}(N_{u},0,\dots,0)_\alpha\right) \mod n}$, where $(N_{u},0,\dots,0)_\alpha$ stands for the $\alpha$-entry in the vector $(N_{u},0,\dots,0)$. Therefore, for each $u\in \{v'\}\cup h^+[v']$, we define $\vec{g}'(u)=(N_{u},0,\dots,0)$. Again, for $\vec{g}'$ to be well-defined it is crucial that $R_X$ is irreflexive, for this condition implies that $h^-[v']$, $\{v'\}$, and $h^+[v']$ are pairwise disjoint sets.  Observe that the definition of $\vec{g}'$ implies that $\left(\overrightarrow{d}', \vec{g}',v'\right)\in W^M$, and that the fact that $w'\approx_\alpha v'$ implies by definition of $\approx_\alpha^M$ that $\left(\overrightarrow{c}', \vec{f}', w'\right)\approx_\alpha^M \left(\overrightarrow{d}', \vec{g}', v'\right)$. By construction, we also have that $\left(\overrightarrow{d}', \vec{g}',v'\right)\in \overline{\left(\overrightarrow{c}_2, \vec{f}_2, w_2\right)}$, so we have finished. 

\item $(\mathtt{NoF})_K$ Let $\alpha \in Ags$ and $\left(\overrightarrow{c}, \vec{f}, w\right),\left(\overrightarrow{d}, \vec{g}, v\right)\in W^M$. Assume that $\left(\overrightarrow{c}, \vec{f}, w\right)\approx_\alpha^M\circ R_X^M \left(\overrightarrow{d}, \vec{g}, v\right) $. This means that $\left(\overrightarrow{c}_{w^{+1}}, \vec{f}, w^{+1}\right)\approx_\alpha^M\left(\overrightarrow{d}, \vec{g}, v\right)$. By definition of $\approx_\alpha^M$, this implies that $w^{+1}\approx_\alpha v$. Since $\mathcal{M}$ validates $(\mathtt{NoF})_K$, we get that this implies that $w\approx_\alpha v^{-1}$. Therefore, the definition of $\approx_\alpha^M$ yields that $\left(\overrightarrow{c}, \vec{f}, w\right)\approx_\alpha^M \left(\overrightarrow{d}_{v^{-1}}, \vec{g}, v^{-1}\right)$.

\end{enumerate}

\end{enumerate}

\end{proof}

\begin{proposition}\label{hetlab}
If $\mathcal{M}$ is a super-additive $n$-model for $\mathcal{L}_{\textsf{KX}}$ where the `next' and `last' relations are irreflexive, then $f:\mathcal{M}^M\to \mathcal{M}$ defined by $f\left(\left(\overrightarrow{c}, \vec{f}, w\right)\right)=w$ is a surjective bounded morphism. 
\end{proposition}
\begin{proof}
\begin{itemize}
\item It is clear that $f$ is surjective. The definition of $\mathcal{V}^M$ in the last item of Definition \ref{matriz} ensures that for every $\left(\overrightarrow{c}, \vec{f}, w\right)\in W^M$, $\left(\overrightarrow{c}, \vec{f}, w\right)$ and $f\left(\left(\overrightarrow{c}, \vec{f}, w\right)\right)$ satisfy the same propositional letters.
    \item Let $R\in\{R_X, R_Y, R_\square, R_\alpha, R_{Ags}, \approx_\alpha\}$, and let $R^M$ stand for the corresponding  relation on $\mathcal{M}^M$. Definition \ref{matriz} ensures that if $\left(\overrightarrow{c}, \vec{f}, w\right) R^M \left(\overrightarrow{d}, \vec{g}, v\right)$, then $w R v$.
    \item Assume that $f\left(\left(\overrightarrow{c}, \vec{f}, w\right)\right) R \ v$ for some $R\in\{R_X, R_Y, R_\square, R_\alpha, R_{Ags}, \approx_\alpha\}$ and $v\in W$. We have the following cases: \begin{itemize}
        \item (Case $R\in\{R_X,R_Y\}$) We have that $\left(\overrightarrow{c}_{v}, \vec{f}, v\right)$ is such that \\
        $\left(\overrightarrow{c}, \vec{f}, w\right)R^M\left(\overrightarrow{c}_{v}, \vec{f}, v\right)$.
        \item (Case $R=R_\square$). Assume that $f\left(\left(\overrightarrow{c}, \vec{f}, w\right)\right) R_\square v$. This implies that $wR_\square v$. We have that $\left(\overrightarrow{c}_{v}, \vec{g}, v\right)$ is such that
        $\left(\overrightarrow{c}, \vec{f}, w\right)R_\square^M\left(\overrightarrow{c}_{v}, \vec{g}, v\right)$, where we define $\vec{g}$ by parts in the following way. For every $u\in h^-[v]$, we know by Lemma \ref{func} that there exists $u'\in h^-[w]$ such that $uR_{Ags}u'$. Therefore, for every $u\in h^-[v]$, we take $\vec{g}(u)= \vec{f}(u')$. Observe that this part of $\vec{g}$ is well-defined on account of the fact that for every $x,y\in h^-[w]$ such that $xR_{Ags}y$ we have that $\vec{f}(x)=\vec{f}(y)$. For each $u\in \{v\}\cup h^+[v]$, we take $\overrightarrow{d}_{u}$ to be the unique choice profile in $\mathcal{C}$ such that $u\in \bigcap\limits_{\alpha\in Ags}\left(\overrightarrow{d}_{u}\right)_\alpha$,\footnote{If $u=v$, then $\overrightarrow{d}_{u}=\overrightarrow{c}_v$.} and we consider the elements in $A_{\overrightarrow{d}_{u}}$. It is clear that for each $u\in \{v\}\cup h^+[v]$, there is a unique element in $A_{\overrightarrow{d}_{u}}$ that includes $u$. Let $N_{u}$ be the index of said element in the enumeration $\left\{A_{\overrightarrow{d}_{u}}^0,\dots, A_{\overrightarrow{d}_{u}}^{n-1}\right\}$. If we take $(N_{u},0,\dots,0)\in \prod\limits_{\alpha\in Ags}\{0,\dots,n-1\}$, then it is clear that $u\in A_{\overrightarrow{d}_{u}}^{\left(\sum_{\alpha\in Ags}(N_{u},0,\dots,0)_\alpha\right) \mod n}$, where $(N_{u},0,\dots,0)_\alpha$ stands for the $\alpha$-entry in the vector $(N_{u},0,\dots,0)$. Therefore, for each $u\in \{v\}\cup h^+[v]$, we define $\vec{g}(u)=(N_{u},0,\dots,0)$. Again, for $\vec{g}$ to be well-defined it is crucial that $R_X$ is irreflexive, for this condition implies that $h^-[v]$, $\{v\}$, and $h^+[v]$ are pairwise disjoint sets.  Observe that the definition of $\vec{g}$ implies that $\left(\overrightarrow{c}_v, \vec{g},v\right)\in W^M$, so that the fact that $wR_\square v$ implies by definition that $\left(\overrightarrow{c}, \vec{f}, w\right)R_\square^M \left(\overrightarrow{c}_v, \vec{g}, v\right)$.
        
        \item (Case $R=R_\alpha$). Fix $\alpha\in Ags$ and assume that $f\left(\left(\overrightarrow{c}, \vec{f}, w\right)\right) R_\alpha v$. This implies that $wR_\alpha v$, so that $\left(\overrightarrow{c}\right)_\alpha=\left(\overrightarrow{c}_{v}\right)_\alpha.$ We have that $\left(\overrightarrow{c}_{v}, \vec{g}, v\right)$ is such that
        $\left(\overrightarrow{c}, \vec{f}, w\right)R_\alpha^M\left(\overrightarrow{c}_{v}, \vec{g}, v\right)$, where we define $\vec{g}$ by parts in the following way. For every $u\in h^-[v]$, we define $\vec{g}$ exactly as in the above item. For $v$, we have to be careful. Consider the elements in $A_{\overrightarrow{c}_{v}}$. It is clear that there is a unique element in $A_{\overrightarrow{c}_{v}}$ that includes $v$. Let $N_{v}$ be the index of said element in the enumeration $\left\{A_{\overrightarrow{c}_{v}}^0,\dots, A_{\overrightarrow{c}_{v}}^{n-1}\right\}$. Let $M\in \mathds{N}\cup \{0\}$ be such that $\left(M +\left(\vec{f}(w)\right)_\alpha\right)\mod n =N_v$, and let  $\alpha_1,\dots,\alpha_m$ be an enumeration of $Ags$ such that $\alpha= \alpha_j$,  $\left(m_{\alpha_1},\dots,m_{\alpha_m}\right)\in \prod\limits_{\alpha\in Ags}\{0,\dots,n-1\}$, and such that $m_{\alpha_j}= \left(\vec{f}(w)\right)_\alpha$ and $\sum_{i\neq j}m_{\alpha_i}=M$.  Then it is clear that $\left(\sum_{\alpha\in Ags}\left (m_{\alpha_1},\dots,m_{\alpha_m}\right)_\alpha\right)\mod n=N_v$. 
        
        Therefore, we set $\vec{g}(v)=\left (m_{\alpha_1},\dots,m_{\alpha_m}\right)$.

       For each $u\in h^+[v]$, we take $\overrightarrow{d}_{u}$ to be the unique choice profile in $\mathcal{C}$ such that $u\in \bigcap\limits_{\alpha\in Ags}\left(\overrightarrow{d}_{u}\right)_\alpha$ and we consider the elements in $A_{\overrightarrow{d}_{u}}$. It is clear that for each $u\in h^+[v]$, there is a unique element in $A_{\overrightarrow{d}_{u}}$ that includes $u$. Let $N_{u}$ be the index of said element in the enumeration $\left\{A_{\overrightarrow{d}_{u}}^0,\dots, A_{\overrightarrow{d}_{u}}^{n-1}\right\}$. If we take $(N_{u},0,\dots,0)\in \prod\limits_{\alpha\in Ags}\{0,\dots,n-1\}$, then it is clear that $u\in A_{\overrightarrow{d}_{u}}^{\left(\sum_{\alpha\in Ags}(N_{u},0,\dots,0)_\alpha\right) \mod n}$, where $(N_{u},0,\dots,0)_\alpha$ stands for the $\alpha$-entry in the vector $(N_{u},0,\dots,0)$. Therefore, for each $u\in h^+[v]$, we define $\vec{g}(u)=(N_{u},0,\dots,0)$. Once again, $\vec{g}$ is well-defined because $R_X$ is irreflexive, for this condition implies that $h^-[v]$, $\{v\}$, and $h^+[v]$ are pairwise disjoint sets.  Observe that the definition of $\vec{g}$ implies that $\left(\overrightarrow{c}_v, \vec{g},v\right)\in W^M$, so that the facts that $wR_\alpha v$ and $\left(\overrightarrow{c}\right)_\alpha=\left(\overrightarrow{c}_{v}\right)_\alpha$  imply that $\left(\overrightarrow{c}, \vec{f}, w\right)R_\alpha^M \left(\overrightarrow{c}_v, \vec{g}, v\right)$.
    \item (Case $R=R_{Ags}$) In this case, assume that $f\left(\left(\overrightarrow{c}, \vec{f}, w\right)\right) R_{Ags} v$. This implies that $wR_{Ags}v$, so that $\overrightarrow{c}=\overrightarrow{c}_{v}.$ We have that $\left(\overrightarrow{c}_{v}, \vec{g}, v\right)$ is such that
        $\left(\overrightarrow{c}, \vec{f}, w\right)R_{Ags}^M\left(\overrightarrow{c}_{v}, \vec{g}, v\right)$, where we define $\vec{g}$ by parts in the following way. For every $u\in h^-[v]\cup h^+[v]$, we define $\vec{g}$ exactly as in the above items. For $v$, we define $\vec{g}(v)=\vec{f}(w)$. Observe that the definition of $\vec{g}$ implies that $\left(\overrightarrow{c}, \vec{f}, w\right)R_{Ags}^M \left(\overrightarrow{c}_v, \vec{g}, v\right)$ in an analogous way to the above items. 
        \item (Case $R=\approx_\alpha$) We have that $\left(\overrightarrow{c}, \vec{f}, w\right)\approx_\alpha^M\left(\overrightarrow{c}_{v}, \vec{g}, v\right)$ for any $\vec{g}$ that fulfills the requirements in Definition \ref{matriz} such that $\left(\overrightarrow{c}_{v}, \vec{g}, v\right)\in W^M$. One such $\vec{g}$ exists in virtue of an argument similar to the ones we used above to build $\vec{g}$ over $h^+[v]$. 

    \end{itemize}
\end{itemize}
 Therefore, $f:\mathcal{M}^M\to \mathcal{M}$ is a surjective bounded morphism. 
\end{proof}

\begin{proposition}\label{compumod}
The proof system $\Lambda_n$ is complete with respect to the class of Kripke-exstit $n$-models.
\end{proposition}
\begin{proof}
Let $\phi$ be a $\Lambda$-consistent formula of $\mathcal{L}_{\textsf{KX}}$. By Proposition \ref{compunr}, we know that there exists a  Kripke-exstit super-additive  $n$-model where the `next' and `last' relations are irreflexive $\mathcal{M}$ and a world $w$ in its domain such that $\mathcal{M},w\models \phi$. By Proposition \ref{hetlab} and the invariance of modal satisfaction under bounded morphisms, we then have that $\mathcal{M}^M$ --as defined in Definition \ref{matriz}-- is a such that  $\mathcal{M}^M,(\overrightarrow{c}_w,\vec{f},w) \models \phi$, where by Proposition \ref{putamadreojala} we know that $\mathcal{M}^M$ is a Kripke-exstit $n$-model. 
\end{proof}



\end{subappendices}

\clearpage

\end{document}